\documentclass[UKenglish]{lipics-v2016-nodoi}
\usepackage{wrapfig}

\ArticleNo{1}

\RequirePackage{bbold}
\RequirePackage{array}
\RequirePackage{xspace}
\RequirePackage{xcolor}
\RequirePackage{hyperref}
\hypersetup{pdfborder={0 0 0}}
\RequirePackage{url}
\IfFileExists{microtype.sty}{\RequirePackage{microtype}}{}
\RequirePackage{amsmath,amsfonts,amssymb,dsfont,mathtools,stmaryrd}
\SetSymbolFont{stmry}{bold}{U}{stmry}{m}{n}
\DeclareFontFamily{U}{stmry}{}
\DeclareFontShape{U}{stmry}{m}{n}
   {  <5> <6> <7> <8> <9> <9.5> <10> gen * stmary
      <10.95><12><14.4><17.28><20.74><24.88>stmary10
   }{}

\makeatletter

\theoremstyle{plain}
   
   \newtheorem{example}{Example}

\def\subexample#1{\expandafter\ifx\csname c@ex@#1\endcsname\relax
  \newcounter{ex@#1}%
  \newcounter{subex@#1}%
  \stepcounter{example}
  \setcounter{ex@#1}{\c@example}%
  \setcounter{subex@#1}{1}%
  \let\save@theexample\theexample
  \@ifundefined{theHexample}\relax
    {\let\save@theHexample\theHexample}
  \else
  \stepcounter{subex@#1}%
  \fi
  \addtocounter{example}{-1}%
  \def\theexample{\arabic{ex@#1}.\alph{subex@#1}}%
  \@ifundefined{theHexample}\relax
    {\def\theHexample{\arabic{ex@#1}.\alph{subex@#1}}}
  \begin{example}}
\def\endsubexample{\end{example}
  \let\theexample\save@theexample
  \@ifundefined{theHexample}\relax
    {\let\theHexample\save@theHexample}}

\newif\ifNM@ux
\def\NM@tuple#1,#2,#3,#4\@NMendtuple{\def\NM@rgtwo{#2}%
  \def\NM@ux{#3}\def\NM@@ux{\relax}
  #1
  \ifx\NM@rgtwo\NM@@ux\else,\ifx\NM@ux\NM@@ux\else
      \ifNM@ux\NM@uxfalse\else\penalty-1000\fi\fi 
    \NM@tuple #2,#3,#4\@NMendtuple
  \fi}
\def\tuple#1{\NM@uxtrue\langle \NM@tuple 
  #1,\relax,\relax,\relax\@NMendtuple \penalty10000\rangle}

\IfFileExists{trimspaces.sty}
  {\RequirePackage{trimspaces}}
  {}
\catcode`\Q=3
\@ifundefined{trim@spaces}
   {
    \newcommand\trim@spaces[1]{%
     \romannumeral-`\q\trim@trim@\noexpand#1Q Q%
    }
    \long\def\trim@trim@#1 Q{\trim@trim@@#1Q}
    \long\def\trim@trim@@#1Q#2{#1}}
   {}
\catcode`\Q=11
\RequirePackage{xkeyval}
\define@boolkeys{mdef}{noerr,nowarn}[true]
\ifKV@mdef@noerr
\presetkeys{mdef}{noerr=true}{}
\else
\presetkeys{mdef}{noerr=false}{}
\fi
\ifKV@mdef@nowarn
\presetkeys{mdef}{nowarn=true}{}
\else
\presetkeys{mdef}{nowarn=false}{}
\fi
\define@key{mdef}{prefix}{\def\@mdprefix{#1}}
\define@key{mdef}{p}{\def\@mdprefix{#1}}
\define@key{mdef}{suffix}{\def\@mdsuffix{#1}}
\define@key{mdef}{s}{\def\@mdsuffix{#1}}
\define@key{mdef}{arg}{\def\@mdargs{#1}}
\define@key{mdef}{args}{\def\@mdargs{#1}}
\define@boolkeys{mdef}{long,global,robust}[true]
\presetkeys{mdef}
           {p=,s=,prefix=,suffix=,long=false,global=false,robust=false,
            arg=0,args=0}{}
\def\@mdef@az{a-z}
\def\@mdef@AZ{A-Z}
\def\@mdef@alphabet{a,b,c,d,e,f,g,h,i,j,k,l,m,n,o,p,q,r,s,t,u,v,w,x,y,z}
\def\@mdef@Alphabet{A,B,C,D,E,F,G,H,I,J,K,L,M,N,O,P,Q,R,S,T,U,V,W,X,Y,Z}
\newcommand\multidef[3][]{%
  \setkeys{mdef}{#1}%
  \def\@mdef@com##1{#2}%
  \@mdef#3,\@end}
\def\@mdef #1,#2\@end{%
  \edef\@mdef@arg{\trim@spaces{#1}}%
  \ifx\@mdef@arg\@mdef@az
    \expandafter\@mdef \@mdef@alphabet,\@end
  \else
    \ifx\@mdef@arg\@mdef@AZ
      \expandafter\@mdef \@mdef@Alphabet,\@end
    \else
      \expandafter\@@mdef\@mdef@arg->->->\@end
    \fi
  \fi
  \def\@mdef@arg{#2}%
  \ifx\@mdef@arg\@empty\else\@mdef #2\@end\fi}
\newtoks\@mdef@redeftok
\def\@mdef@comma{}
\def\@@mdef#1->#2->#3\@end{%
  \@ifundefined{\@mdprefix#1\@mdsuffix}
    {\@@@mdef{#1}{#2}}
    {\ifKV@mdef@nowarn\else
       \edef\@mdef@redef{\the\@mdef@redeftok\@mdef@comma
         \@backslashchar\@mdprefix#1\@mdsuffix}
       \def\@mdef@comma{, }
       \global\@mdef@redeftok=\expandafter{\@mdef@redef}
     \fi
     \ifKV@mdef@noerr
       \@@@mdef{#1}{#2}%
       \ifKV@mdef@nowarn\else
         \PackageWarning{multidef}
           {command \expandafter\noexpand\csname\@mdprefix#1\@mdsuffix\endcsname
             redefined}
       \fi
     \else
       \PackageError{multidef}
         {command \expandafter\noexpand\csname\@mdprefix#1\@mdsuffix\endcsname
           already defined}\@ehc
     \fi
     \ifKV@mdef@nowarn\else
       \@ifundefined{@mdwarnonce}
         {\def\@mdwarnonce{}%
          \@mdef@finalwarn}
         {}
     \fi}
}
\def\@mdef@finalwarn{%
  \AtEndDocument{\PackageWarningNoLine{multidef}{There were
     redefined commands (\the\@mdef@redeftok)}}}
\def\@@@mdef#1#2{\def\@arg@{#2}%
  \ifx\@arg@\@empty
    \ifKV@mdef@robust
      \expandafter\def\expandafter\@mdef@cmdname
        \expandafter{\csname\@mdprefix#1\@mdsuffix\endcsname}%
      \expandafter\@mdef@robdef\@mdef@cmdname{#1}%
    \else
      \@mdef@def{#1}{#1}%
    \fi
  \else
    \ifKV@mdef@robust
      \expandafter\def\expandafter\@mdef@cmdname
        \expandafter{\csname\@mdprefix#1\@mdsuffix\endcsname}
      \expandafter\@mdef@robdef\@mdef@cmdname{#2}%
    \else
      \@mdef@def{#1}{#2}%
    \fi
  \fi}
\def\@mdef@def#1#2{%
  \let\reserved@b\@gobble
  \ifKV@mdef@global\let\@mdglobal\global\else\let\@mdglobal\relax\fi
  \ifKV@mdef@long\let\@mdlong\long\else\let\@mdlong\relax\fi
  \def\l@ngrel@x{\@mdlong\@mdglobal}
  \expandafter\expandafter\expandafter\@yargd@f\expandafter\@mdargs\csname
  \@mdprefix#1\@mdsuffix\expandafter\endcsname\expandafter{\@mdef@com{#2}}
}
\def\@mdef@robdef#1#2{%
  \edef\reserved@a{\string#1}%
  \def\reserved@b{#1}%
  \edef\reserved@b{\expandafter\strip@prefix\meaning\reserved@b}%
  \global\edef#1{%
     \ifx\reserved@a\reserved@b
        \noexpand\x@protect
        \noexpand#1%
     \fi
     \noexpand\protect
     \expandafter\noexpand\csname
        \expandafter\@gobble\string#1 \endcsname
  }%
  \let\reserved@b\@gobble
  \ifKV@mdef@global\let\@mdglobal\global\else\let\@mdglobal\relax\fi
  \ifKV@mdef@long\let\@mdlong\long\else\let\@mdlong\relax\fi
  \def\l@ngrel@x{\@mdlong\@mdglobal}
  \expandafter\expandafter\expandafter\@yargd@f\expandafter\@mdargs\csname
    \expandafter\@gobble\string#1 \expandafter\endcsname
    \expandafter{\@mdef@com{#2}}
}

\multidef[p=cal]{\ensuremath{\mathcal{#1}}\xspace}{A-Z}
\newcommand\optbb[1]{%
  \@ifnextchar+{\ensuremath{\mathbb{#1}_{\geq 0}}\afterassignment\xspace\let\next=}
               {\@ifnextchar*{\ensuremath{\mathbb{#1}_{>0}}\afterassignment\xspace\let\next=}
                             {\ensuremath{\mathbb{#1}}\xspace}}}
\multidef[p=bb]{\optbb{#1}}{A,B,C,H,N,Q,R,T,Z,One->1,one->1}
\multidef{\textit{#1}\xspace}{init,start,aux,sent,,halt,itend->end,token->tok,
  sink,idle}

\RequirePackage{xspace}
\def\ComplexityFont#1{\ensuremath{\mathsf{#1}}\xspace}

\multidef[noerr]{\ComplexityFont{#1}}{LOGDCFL,PTIME,NP,PH,NLOGSPACE,LOGSPACE,PSPACE,%
  NPSPACE,@EXPTIME->EXPTIME,@EXPSPACE->EXPSPACE,NEXPSPACE}

\newcommand\EXPTIME[1][]{\bgroup\def\@aux{#1}%
  \ifx\@aux\@empty\egroup\else\egroup\ensuremath{\mathsf{\mathit{#1}}}\hbox{-}\fi\@EXPTIME}
\newcommand\EXPSPACE[1][]{\def\@aux{#1}%
  \ifx\@aux\@empty\else\ensuremath{\mathsf{\mathit{#1}}}\hbox{-}\fi\@EXPSPACE}

\usepackage{tikz}
\usetikzlibrary{calc,shapes,arrows,chains,positioning,automata, decorations.pathreplacing}
\tikzset{rn/.style={circle,draw, node distance=2cm}}
\tikzset{final/.style={rectangle,draw, node distance=2cm}}
\tikzset{annotate/.style={node distance=0.9cm}}
\tikzset{entry/.style={initial by arrow, initial text=, initial where=above}}
\colorlet{drouge}{red}
\colorlet{frouge}{red!20!white}
\colorlet{dbleu}{blue}
\colorlet{fbleu}{blue!30!white}
\colorlet{fbleuc}{blue!20!white}
\colorlet{dviolet}{blue!50!red}
\colorlet{fviolet}{blue!50!red!40!white}
\colorlet{dvert}{green!80!black}
\colorlet{fvert}{green!80!black!20!white}
\colorlet{djaune}{yellow!80!black}
\colorlet{fjaune}{yellow!80!black!20!white}
\colorlet{dgris}{white!60!black}
\colorlet{fgris}{white!90!black}
\colorlet{dgrisf}{white!30!black}
\colorlet{fgrisf}{white!70!black}

\tikzstyle{ptrond}=[draw,circle,minimum height=2mm]
\tikzstyle{ptcarre}=[draw,minimum width=3mm,minimum height=3mm]
\tikzstyle{moyrond}=[draw,circle,minimum height=5mm]
\tikzstyle{moycarre}=[draw,minimum width=4mm,minimum height=4mm]
\tikzstyle{rond}=[draw,circle,minimum height=7mm]
\tikzstyle{carre}=[draw,minimum width=6mm,minimum height=6mm]
\tikzstyle{rouge}=[draw=drouge,fill=frouge]
\tikzstyle{vert}=[draw=dvert,fill=fvert]
\tikzstyle{jaune}=[draw=djaune,fill=fjaune]
\tikzstyle{bleu}=[draw=dbleu,fill=fbleu]
\tikzstyle{violet}=[draw=dviolet,fill=fviolet]
\tikzstyle{gris}=[draw=dgris,fill=fgris]
\tikzstyle{grisf}=[draw=dgrisf,fill=fgrisf]
\tikzstyle{rvert}=[style=rond,style=vert]
\tikzstyle{rrouge}=[style=rond,style=rouge]

\newcommand\upcl[1][]{\def\@tempa{#1}\ifx\@tempa\@empty\relax
  \ensuremath{\mathord{\uparrow}}\else
  \ensuremath{\mathord{\uparrow_{#1}}}\fi}
\newcommand\support[2][]{\ensuremath{\overline{\vphantom{#1}#2}\xspace}}
\newcommand{\pr}{\mathbb{P}}

\newcommand{\Pre}{\mathrm{Pre}}

\def\pinit{\ensuremath{p_{\textsf{init}}\xspace}}
\def\psink{\ensuremath{p_{\textsf{sink}}\xspace}}

\def\qhalt{\ensuremath{q_{\textsf{halt}}\xspace}}

\let\indic\bbone

\makeatother

\let\Prot\calP

\let\Sys\calS
\let\SG\calG

\let\Mcup\oplus
\let\Msetminus\ominus
\let\Msubseteq\sqsubseteq
\def\Msize#1{|#1|}

\newcommand{\Prob}{\mathit{Pr}}
\let\Confs\Gamma
\let\SubConfs\Delta
\newcommand{\ConfsF}{\llbracket q_f \rrbracket}
\let\aconf\gamma
\newcommand{\Objective}[1]{\llbracket \Diamond #1 \rrbracket}
\newcommand{\state}[1]{st(#1)}
\newcommand{\data}[1]{data(#1)}
\newcommand{\set}[1]{\{#1\}}

\newcommand{\Post}{\mathrm{Post}}

\author[1]{Patricia~Bouyer}
\author[1]{Nicolas~Markey}
\author[2]{Mickael~Randour\thanks{F.R.S.-FNRS Postdoctoral Researcher.}}
\author[3]{Arnaud~Sangnier}
\author[1]{Daniel~Stan}
\authorrunning{Patricia~Bouyer, Nicolas~Markey, Mickael Randour, Arnaud Sangnier and Daniel~Stan}

\Copyright{Patricia~Bouyer,  Nicolas~Markey, Mickael Randour, Arnaud Sangnier and Daniel~Stan}

\affil[1]{LSV -- CNRS, ENS Cachan \& University Paris-Saclay -- France}
\affil[2]{Computer Science Department -- Universit\'e Libre de Bruxelles (ULB) --
  Belgium} 
\affil[3]{IRIF -- University Paris Diderot \& CNRS -- France}

\title{Reachability in Networks of Register Protocols under Stochastic
  Schedulers\footnote{This paper extends the conference version presented in~\cite{BMR+16b}.}\footnote{This work has been partly supported by ERC Starting grant EQualIS
(FP7-308087), by European FET project Cassting (FP7-601148), and  by
the  ANR  research  program PACS (ANR-14-CE28-0002).}} 
\titlerunning{Reachability in Networks of Register Protocols under Stochastic
  Schedulers}

\begin{document}
\hfuzz=0pt
\maketitle

\begin{abstract}
  We study the almost-sure reachability problem in a distributed system
  obtained as the asynchronous composition of \(N\) copies (called
  processes) of the same automaton (called protocol), that can
  communicate via a shared register with finite domain. The automaton
  has two types of transitions: write-transitions update the value of
  the register, while read-transitions move to a new state depending on
  the content of the register. Non-determinism is resolved by a
  stochastic scheduler. Given a protocol, we focus on almost-sure
  reachability of a target state by one of the processes. The answer to
  this problem naturally depends on the number~\(N\) of processes. However,
  we prove that our setting has a cut-off property: the~answer to the
  almost-sure reachability problem is constant when \(N\) is large enough;
  we~then develop an \EXPSPACE algorithm deciding whether this constant
  answer is positive or negative.
\end{abstract}

\section{Introduction}

\subparagraph{Verification of systems with many identical processes.} 
It is a classical pattern in distributed systems to have a large number of
identical components running concurrently (a.k.a. networks of
processes). In~order to verify the correctness of such systems, a~naive option
consists in fixing an upper bound on the number of processes, and applying
classical verification techniques on the resulting system. This has several
drawbacks, and in particular it gives no information whatsoever about larger
systems. Another option is to use parameterized-verification techniques, taking as a
parameter the number of copies of the protocol in the system being considered.
In~such a setting, the natural question is to find and characterize the
set of parameter values for which the system is 
correct. Not~only the latter approach is more general, but it might also turn out to be
easier and more efficient, since it involves orthogonal techniques.

\subparagraph{Different means of communication lead to different models.}
A~seminal paper on parameterized verification of such distributed systems is
the work of German and Sistla~\cite{GS92}. In~this work, the authors consider
networks of processes all following the same finite-state automaton; the
communication between processes is performed thanks to \emph{rendez-vous}
communication. Various related settings have been proposed and studied since then, 
which mainly differ by the way the processes communicate. Among those,
let us mention broadcast communication~\cite{EFM99,DSZ10},
token-passing~\cite{CTTV04,AJKR14}, message passing~\cite{BGS14},
shared register with ring topologies~\cite{ABG15}, or shared
memory~\cite{EGM13}. In his nice survey on such parameterized models~\cite{Esp14}, Esparza shows that minor changes in the setting, such as the presence of a
controller in the system, might drastically change the complexity 
of the verification problems. The~relative expressiveness of some of those
models has been studied recently in~\cite{ARZ15}, yielding several reductions of
the verification problems for some of those classes of models.

\subparagraph{Asynchronous shared-memory systems.} 
We consider a communication model where the processes 
asynchronously access a shared 
register,  and where read and write operations on this register are performed
non-atomically. A~similar model has been proposed by Hague
in~\cite{Hag11}, 
where the behavior of processes is defined by a pushdown automaton. The~complexity of some reachability and liveness problems for shared-memory models have then been
established in~\cite{EGM13} and~\cite{DEGM15}, respectively. These works consider networks in which a specific process, called the
leader, runs a different program, and address the problem whether, for 
some number of processes, the leader can satisfy a given reachability or
liveness property. In~the case where there is no leader, and where processes are
finite-state, the parameterized control-state reachability problem
(asking whether one of the processes can reach a given control state) can be
solved in polynomial time, by adapting the approach of~\cite{DSTZ12} for lossy
broadcast protocols.

\subparagraph{Fairness  and cut-off properties.} 
\looseness=-1
In this work, we further insert fairness assumptions in the model of
parameterized networks with asynchronous shared memory, 
and consider reachability problems in this setting.
There are different ways to include fairness in
parameterized models. 
One~approach is to enforce fairness expressed as a temporal-logic properties
on the executions (e.g., any action that is available infinitely often must be
performed infinitely often); this~is the option chosen for parameterized
networks with rendez-vous~\cite{GS92} and for systems with disjunctive guards
(where processes can query the states of other processes)
in~\cite{AJK16}. 
We~follow another choice, by equipping our networks with a stochastic
scheduler that, at~each step of the execution, 
assigns the same probability to the available actions of all the processes. From~a high-level perspective, both forms of fairness are
  similar.  However, expressing
  fairness via temporal logic allows for very regular patterns (e.g.,
  round-robin execution of the processes), whereas the stochastic approach
  leads to consider all possible interleavings with probability~$1$. Under this stochastic scheduler assumption, we~focus on almost-sure reachability of a given control
state by any of the processes of the system.
More specifically, as~in~\cite{AJK16}, we~are interested in determining the
existence of a \emph{cut-off}, i.e.,~an~integer~$k$ such that networks with
more than $k$ processes almost-surely reach the target state. Deciding the existence and computing such cut-offs is important for at~least two
aspects: first, it~ensures that the system is 
correct for arbitrarily large networks; second, if we~are able to derive a bound
on the cut-off, then using classical verification techniques we can find the
exact value of the cut-off and exactly characterize the sizes of the networks
for which the behavior is correct.

\subparagraph{Our contributions.}  
We prove that for finite-state asynchronous shared-memory protocols with a
stochastic scheduler, and for almost-sure reachability of some control state
by some process of the network, there always exists a positive or negative
cut-off; positive cut-offs are those above which the target state is reached
with probability~$1$, while negative cut-offs are those above which the target
state is reached with probability strictly less than~$1$. Notice that both
cut-offs are not complement of one another, so that our result is not trivial. 

We~then prove that the ``sign'' (positive or negative) of a cut-off can be
decided in \EXPSPACE, and that this problem is \PSPACE-hard.
Finally, we~provide lower and upper bounds on the values of the cut-offs,
exhibiting in particular protocols with exponential (negative) cut-off.
Notice how these results contrast with classical results in related areas: in
the absence of fairness, reachability can be decided in polynomial time, and
in most settings, when cut-offs exist, they generally have polynomial
size~\cite{AJK16,EN03,EK00}.

\section{Presentation of the model and of the considered problem}

\subsection{Preliminaries.}

Let $S$ be a finite set. A~multiset over~$S$ is a mapping $\mu\colon S \to
\bbN$. 
The~cardinality of a multiset~$\mu$ is $\Msize\mu=\sum_{s\in S}\mu(s)$.
The~support~$\support\mu$ of~$\mu$ is the~subset $\nu\subseteq S$ s.t.
for all~$s\in S$, it~holds $s\in \nu$ if, and only~if, $\mu(s)>0$.
For~$k\in\bbN$, we~write $\bbN^S_k$ for the set of multisets of
cardinality~$k$ over~$S$,
and~$\bbN^S$ for the set of all multisets over~$S$.
For any~$s\in S$ and~$k\in\bbN$, we~write $s^k$ for the multiset where
$s^k(s)=k$ and $s^k(s')=0$ for all $s'\not=s$.
We~may write $s$ instead of~$s^1$ when no ambiguity may arise.
A~multiset~$\mu$ is included in a multiset~$\mu'$, written
$\mu\Msubseteq \mu'$, if $\mu(s)\leq
\mu'(s)$ for all~$s\in S$.
Given two multisets~$\mu$ and~$\mu'$, their union $\mu\Mcup\mu'$ is still a
multiset~s.t. $(\mu\Mcup\mu')(s)=\mu(s)  + \mu'(s)$ for all~$s\in S$.
Assuming $\mu\Msubseteq\mu'$, the difference $\mu'\Msetminus\mu$ is still a
multiset~s.t. $(\mu'\Msetminus\mu)(s)=\mu'(s)-\mu(s)$.

\smallskip

A~quasi-order $\tuple{A,\preceq}$ is a \emph{well quasi-order} (wqo for short)
if for every infinite sequence of elements $a_1,a_2,\ldots$ in~$A$, there
exist two indices $i<j$ such that $a_i \preceq a_j$. For~instance, for $n>0$,
$\tuple{\bbN^n,\leq}$ (with lexicographic order) is a~wqo. Given a set~$A$
with an ordering~$\preceq$ and a subset $B \subseteq A$, the~set~$B$ is said
to be \emph{upward closed} in~$A$ if for all $a_1 \in B$ and $a_2 \in A$, in
case $a_1 \preceq a_2$, then ${a_2 \in B}$. The~\emph{upward-closure} of a
set~$B$ (for~the ordering~$\preceq$), denoted by $\upcl[\preceq](B)$
(or~sometimes $\upcl(B)$ when the ordering is clear from the context), is
the set $\set{a \in A \mid \exists b \in B \mbox{ s.t. } b \preceq a}$. If
$\tuple{A,\preceq}$ is a wqo and $B$~is an upward closed set in~$A$, there
exists a finite set of minimal elements $\set{b_1,\ldots,b_k}$ such that
$B=\upcl{\set{b_1,\ldots,b_k}}$.

\subsection{Register protocols and associated distributed system.}
We focus on systems that are defined as the (asynchronous)
product of several copies of the same protocol. Each copy communicates
with the others through a single register that can store values from a
finite alphabet.

\begin{definition}
    A \emph{register protocol} is given by
    $\Prot = \tuple{Q, D, q_0, T}$, where 
    $Q$ is a finite set of control locations, 
    $D$ is a finite alphabet of data for the shared register, 
    $q_0\in Q$ is an initial location, 
    $T\subseteq Q\times \{R,W\}\times D\times Q$
            is the set of transitions of the protocol.  Here $R$ means \emph{read} the content of the shared register,
            while $W$ means \emph{write} in the register.

    In~order to avoid deadlocks, it~is required that each location has
    at least one outgoing transition. We also require that whenever some $R$-transition $(q,R,d',q')$
    appears in~$T$, then for all~$d\in D$, there exists at least
    one~$q_d\in Q$ such that $(q,R,d,q_d)\in T$. The size of the protocol $\Prot$ is
given by $|Q| + |T|$.
\end{definition}

\begin{subexample}{running}\label{ex1}
Figure~\ref{fig-ex-protocol} displays a small register protocol with four locations, over
an alphabet of data~$D=\{0,1,2\}$. In this figure (and in the sequel), omitted
$R$-transitions (e.g.,~transitions~$R(1)$ and~$R(2)$ from~$q_0$) are assumed to
be self-loops. When the register contains~$0$, this protocol may move from
initial location~$q_0$ to location~$q_1$. From~there it can write~$1$ in the
register, and then move to~$q_2$. From~$q_2$, as long as the register
contains~$1$, the process can either stay in~$q_2$ (with the omitted
self-loop~$R(1)$), or write~$2$ in the register and jump back to~$q_1$. It~is
easily seen that if this process executes alone, it~cannot reach state~$q_f$.
\begin{figure}[t]
\centering
\begin{tikzpicture}
\draw (0,0) node[rond,bleu] (a) {$q_0$};
\draw (2,0) node[rond,vert] (b) {$q_1$};
\draw (4,0) node[rond,jaune] (c) {$q_2$};
\draw (6,0) node[rond,rouge] (d) {$q_f$};
\everymath{\scriptstyle}
\draw[latex'-] (a.180) -- +(180:3mm); 
\draw (a) edge[-latex'] node[below] {$R(0)$} (b);
\draw (b) edge[-latex',out=60,in=120,looseness=6] node[pos=.75,left] {$W(1)$} (b);
\draw (b) edge[-latex',bend right] node[above] {$R(1)$} (c);
\draw (c) edge[-latex',bend right] node[above] {$W(2)$} (b);
\draw (c) edge[-latex'] node[below] {$R(2)$} (d);
\draw (d) edge[-latex',out=60,in=120,looseness=6] node[pos=.75,left] {$W(2)$} (d);
\end{tikzpicture}
\caption{Example of a register protocol with $D=\{0,1,2\}$.}
\label{fig-ex-protocol}
\end{figure}
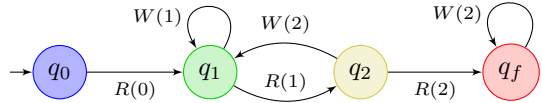
\end{subexample}

We now present the semantics of distributed systems associated with our
register protocols. We~consider the
\emph{asynchronous} composition of several copies of the protocol (the~number
of copies is not fixed a~priori and can be seen as a parameter). We~are
interested in the behavior of such a composition under a fair scheduler. Such
distributed systems involve two sources of non-determinism: first, register
protocols may be non-deterministic; second, in any configuration, all
protocols have at least one available transition, and non-determinism arises
from the asynchronous semantics. In~the semantics associated with a register
protocol, non-determinism will be solved by a randomized scheduler, whose role
is to select at each step which process will perform a transition, and which
transition it will perform among the available ones. Because we will
consider qualitative objectives (almost-sure reachability), the exact probability distributions will not really matter, and we will
pick the uniform one (arbitrary choice).  Note that we assume non-atomic read\slash write operations on the
register, as in~\cite{Hag11,EGM13,DEGM15}.
More precisely, when one process performs a transition, then all the processes
that are in the same state are allowed to also perform the same
transition just after, in fact write are always possible, and if a
process performs a read of a specific value, since this read does not
alter the value of the register, all processes in the same state can
perform the same read (until one process performs a write). We will see later that dropping this hypothesis has a
consequence on our results. We now give the formal definition of such a system.

The configurations of the distributed system built
on register protocol~$\Prot=\tuple{Q, D, q_0, T}$ belong to the set $\Confs=\bbN^Q \times D$. The~first component of a configuration is
a multiset characterizing the number of processes in each state
of~$Q$, whereas the second component provides the content of the register. For~a
configuration~$\aconf=\tuple{\mu,d}$, we~denote by~$\state{\aconf}$ the
multiset~$\mu$ in~$\bbN^Q$ and by~$\data{\aconf}$ the data~$d$ in~$D$.
We~overload the operators defined over multisets; in particular,
for a multiset~$\delta$ over~$Q$, we~write $\aconf\Mcup\delta$ for the
configuration $\tuple{\mu\Mcup\delta,d}$. Similarly, we~write~$\support\aconf$
for the support of~$\state{\gamma}$. 

A~configuration $\aconf'=\tuple{\mu',d'}$ is a \emph{successor} of a
configuration $\aconf=\tuple{\mu,d}$ if, and only~if, there is a transition
$(q,\textsf{op},d'',q')\in T$ such that $\mu(q)>0$, $\mu'=\mu\Msetminus q\Mcup
q'$ and either $\textsf{op}=R$ and $d=d'=d''$, or $\textsf{op}=W$ and
$d'=d''$. In~that case, we~write $\aconf \rightarrow \aconf'$. Note that since
$\mu(q)>0$ and $\mu'=\mu\Msetminus q\Mcup q'$, we have necessarily
$\Msize\mu=\Msize{\mu'}$. In~our system, we~assume that there is no
creation or deletion of processes during an execution, hence
the~size of configurations (i.e.,~$\Msize{\state{\aconf}}$)
remains constant along transitions. We~write $\Gamma_k$ for the set of
configurations of size~$k$.
For any configuration~$\aconf\in\Gamma_k$, we~denote by~$\Post(\aconf)\subseteq
\Gamma_k$ the set of successors  of~$\aconf$, and point out that such a set is
finite and non-empty.

Now, the \emph{distributed system} $\Sys_\Prot$ associated with a register
protocol~$\Prot$ is a discrete-time Markov chain $\tuple{\Confs,\Prob}$ where $\Prob\colon \Gamma \times \Gamma \to [0,1]$ is
the transition probability matrix defined as follows: for~all~$\aconf$
and~$\aconf'\in \Confs$,
we have $\Prob(\aconf,\aconf')=\frac{1}{|\Post(\aconf)|}$ if
$\aconf \rightarrow \aconf'$, and $\Prob(\aconf,\aconf')=0$ otherwise. Note
that $\Prob$ is well defined: by~the restriction imposed on the transition
relation~$T$ of the protocol, we~have
$0<|\Post(\aconf)| < \infty$ for all configuration~$\aconf$, and hence
we also get $\Sigma_{\aconf' \in \Confs}\Prob(\aconf,\aconf')=1$.
For~a fixed integer~$k$, we~define the distributed system of size~$k$
associated with~$\Prot$ as the finite-state discrete-time Markov chain
$\Sys_\Prot^k=\tuple{\Gamma_k,\Pr_k}$, where $\Pr_k$ is the restriction of~$\Pr$
to~$\Gamma_k\times\Gamma_k$. 

\smallskip 

  We are interested in analyzing the behavior of the distributed
  system for a large number of participants. More precisely, we are interested in determining
  whether almost-sure reachability of a specific control state holds
  when the number of processes involved is large. We are therefore
  seeking a \emph{cut-off} property, which we formalize in the
  following.

  A~finite path in the
  system $\Sys_\Prot$ is a finite sequence of configurations $\aconf_0
  \rightarrow \aconf_1 \ldots \rightarrow \aconf_k$. In~such a case, we~say
  that the path starts in~$\aconf_0$ and ends in~$\aconf_k$. We~furthermore
  write $\aconf \rightarrow^\ast \aconf'$ if, and only~if, there exists a path
  that starts in~$\aconf$ and ends in~$\aconf'$. Given a location~$q_f$, we
  denote by $\Objective{q_f}$ the set of paths of the form $\aconf_0
  \rightarrow \aconf_1 \ldots \rightarrow \aconf_k$ for which there is $i \in
  [0;k]$ such that $\state{\aconf_i}(q_f)>0$. Given a configuration~$\aconf$,
  we~denote by $\pr(\aconf,\Objective{q_f})$ the probability that some paths
  starting in $\aconf$ belong to $\Objective{q_f}$ in~$\Sys_\Prot$. This
  probability is well-defined since the set of such paths is measurable (see
  e.g.,~\cite{PoMC2008-BK}). Given a register protocol~$\Prot=\tuple{Q, D, q_0,  T}$, an~initial
register value~$d_0$, and a target location~$q_f\in Q$,  we  say that
$q_f$ is almost-surely reachable for $k$ processes if
$\pr(\tuple{q^k_0,d_0},\Objective{q_f})=1$.

\begin{subexample}{running}
Consider again the protocol depicted in Fig.~\ref{fig-ex-protocol}, with
initial register content~$0$. As~we explained already, for~$k=1$, the final
state is not reachable at~all, for any scheduler (here as~$k=1$, the~scheduler
only has to solve non-determinism in the protocol).

When~$k=2$, one easily sees that the final state is reachable: it~suffices
that both processes go to~$q_2$ together, from where one process may write
value~$2$ in the register, which the other process can read and go to~$q_f$.
Notice that this does not ensure that $q_f$ is reachable almost-surely for
this~$k$ (and actually, it~is~not; see~Example~\ref{ex-1c}).
\end{subexample}

We aim here at finding cut-offs for almost-sure reachability, i.e., we seek
the existence of a threshold such that almost-sure reachability
(or its negation) holds for all larger values.

\begin{definition}
  Fix a protocol~$\Prot=\tuple{Q, D, q_0, T}$, $d_0\in D$, and $q_f\in Q$.
  An integer~$k\in\bbN$ is a \emph{cut-off for almost-sure
    reachability} (shortly a \emph{cut-off}) for~$\Prot$, $d_0$ and $q_f$
  if one of the following two properties holds:
  \begin{itemize}
    \item for all $h\geq k$, we have $\pr(\tuple{q^h_0,d_0},\Objective{q_f})=1$. In~this case $k$ is a \emph{positive}
        cut-off;
    \item for all $h\geq k$, we have
      $\pr(\tuple{q^h_0,d_0},\Objective{q_f})< 1$. Then $k$~is a \emph{negative} cut-off.
  \end{itemize}
An~integer~$k$ is a \emph{tight} cut-off if it is a cut-off and $k-1$ is~not.
\end{definition}

\looseness=-1
Notice that from the definition, cut-offs need not exist for a given distributed
system.
Our~main result precisely states that cut-offs always exist, and that we
can decide their nature.
\begin{theorem}\label{thm-main}
For any protocol~$\Prot$, any initial register value~$d_0$ and any target
location~$q_f$, there always exists a cut-off for almost-sure
    reachability, whose value is at most
doubly-exponential in the size of~$\calP$. Whether it is a
positive or a negative cut-off can be decided in \EXPSPACE, and
is \PSPACE-hard.
\end{theorem}

\begin{remark}
 When dropping the condition on non-atomic read\slash write operations,
 and allowing transitions with atomic read\slash 
write operations (i.e.,~one~process is ensured to perform a read and a
write operation without to be interrupted by another process), the existence of a cut-off
(Theorem~\ref{thm-main}) is not ensured. This is demonstrated with the
protocol of Fig.~\ref{fig-ex-nocutoff}: one easily checks (e.g.,~inductively on
the number of processes, since processes that end up in~$q_2$ play no role
anymore) that state~$q_f$ is reached with probability~$1$ if, and only~if, the
number of processes is odd.
\end{remark}

\section{Properties of register protocols}

\subsection{Example of a register protocol}
\label{subsec:example}
\begin{figure}[tb]
\begin{minipage}[t]{.34\linewidth}
\centering
\begin{tikzpicture}
\path[use as bounding box] (-2.15,.4) -- (2.4,-1.9);
\draw (0,0) node[rond,bleu] (a) {$q_0$};
\draw (0,-1.5) node[rond,jaune] (b) {$q_1$};
\draw (2,-.75) node[rond,rouge] (c) {$q_2$};
\draw (-1.75,0) node[rond,vert] (d) {$q_f$};
\draw[latex'-] (a.135) -- +(135:3mm);

\everymath{\scriptstyle}
\draw (a) edge[-latex', bend right=20] node[left=-1pt] {$\genfrac{}{}{0pt}{1}{R(0)}{W(1)}$} (b);
\draw (b) edge[-latex', bend right=20] node[right=-1pt] {$\genfrac{}{}{0pt}{1}{R(1)}{W(0)}$} (a);
\draw (a) edge[-latex', bend left=0] node[pos=.2,above right=-2pt] {$\genfrac{}{}{0pt}{1}{R(1)}{W(2)}$} (c);
\draw (b) edge[-latex'] node[pos=.2,below right=-2pt] {$\genfrac{}{}{0pt}{1}{R(2)}{W(0)}$} (c);
\draw (a) edge[-latex'] node[above] {$R(0)$} (d);

\draw (c) edge[-latex',out=-60,in=-120,looseness=6] (c);
\draw (d) edge[-latex',out=-60,in=-120,looseness=6] (d);
\end{tikzpicture}
\caption{Example of a register protocol with atomic read\slash write
  operations.}
\label{fig-ex-nocutoff}
\end{minipage}\hfill
\begin{minipage}[t]{.6\linewidth}
\hfill
 \begin{tikzpicture}[node distance=1.4cm]
\path[use as bounding box] (-.4,1.15) -- (7.7,-1.15);
   \path (0,0) node[rond,bleu] (q0) {$s_0$};
   \draw[latex'-] (q0.-135) -- +(-135:3mm);
   \node[rond, right of=q0,bleu] (q1) {$s_1$};
   \node[rond, right of=q1,bleu] (q2) {$s_2$};
   \node[right of=q2,node distance=1.4cm] (qmid) {$\dots$};
   \node[rond, right of=qmid, bleu,inner sep=0pt,node distance=1.4cm] (qn1) {$s_{n-1}$};
   \node[rond, right of=qn1, accepting,vert] (qn) {$s_n$};
   \everymath{\scriptstyle}
   \path (q0) edge[-latex', out=-60,in=-120,looseness=7] node[below,pos=.5] {$W(0)$} (q0)
              edge[-latex', bend right=30] node[below=-1pt] {$R(0)$} (q1); 
   \path (q1) edge[-latex', bend right=30] node[below=-1pt] {$W(1)$} (q0)
              edge[-latex', bend right=30] node[below] {$R(1)$} (q2);
   \path (q2) edge[-latex', bend right, in=-120] node[above,pos=.2] {$W(2)$} (q0)
              edge[dashed, out=-30,in=180] node[below,pos=.8] {$R(2)$} +(.8,-.3);
   \path ($(qn1)+(-.8,-.3)$) edge[-latex', dashed,in=-150,out=0] node[below,pos=.2] {$R(n-2)$} (qn1);
   \path (qn1) edge[-latex', bend right=30] node[below] {$R(n-1)$} (qn);
   \path (qn1.160) edge[dashed] node[pos=.9,above right=-1pt] {$W(n-1)$} +(160:1.2cm);
   \path (q0) edge[dashed,latex'-,out=80,in=-160] +(50:1.5cm);
 \end{tikzpicture}
 \caption{A ``filter'' protocol $\calF_{n}$ for $n>0$.}
 \label{filtern}\label{fig-filtern}
\end{minipage}
\end{figure}

\looseness=-1
We illustrate our model with a family of register
protocols~$\left(\calF_n\right)_{n>0}$, depicted in Fig.~\ref{fig-filtern}.
For~a fixed~$n$, protocol~$\calF_n$ has $n+1$ states and $n$ different
data; intuitively, in~order to move from~$s_i$ to~$s_{i+1}$, two processes
are needed: one writes~$i$ in the register and goes back to~$s_0$, and the
second process can proceed to~$s_{i+1}$ by reading~$i$. 
Since backward transitions to $s_0$ are always possible and since
states can always exit $s_0$ by writing a $0$ and reading it afterwards, no
deadlock can ever occur so the main question remains to determine if $s_n$ is
reachable by one of the processes as we increase the number of initial
processes. As~shown in Lemma~\ref{lem-filter}, the~answer is positive:
$\calF_n$~has a tight linear positive cut-off; it~actually behaves like a ``filter'',
that can test if at least $n$ processes are running together. We~exploit this
property later in Section~\ref{sec-hardness}. 

\begin{lemma}
\label{lem-filter}
 Fix~$n\in\bbN$.
 The ``filter'' protocol $\calF_n$, depicted in Fig.~\ref{filtern}, with
 initial register value~$0$ and target location~$s_n$, has a tight positive
 cut-off~equal to $n$.
\end{lemma}

\newcommand{\op}{\textrm{op}}

\begin{proof}
  We~consider the system~$\Sys_{\calF_n}^m$, made of $m$ copies of~$\calF_n$, with
  initial register value~$0$.
  We first prove that any reachable configuration~$\aconf$ satisfies:
  \[
    \forall j\leq m.\ \sum_{k=0}^{j}\state{\gamma}(s_k) \geq j +
    \indic_{\data{\gamma}=j+1} 
  \]
  The~proof is by induction: 
  the~invariant is satisfied by the initial
  configuration~$\aconf_0=\tuple{s_0^m,0}$.  
  Let us now consider the run
  $\aconf_0 \rightarrow^* \aconf \rightarrow \aconf'$, in which $\aconf$
  satisfies the invariant, and with last transition $(q,\op, d,q')\in T$.
  \begin{itemize}
    \item If $(\op,d)=(R,0)$, then $q=s_0$ and~$q'=s_1$. Along that
      transition, the right-hand-side term is unchanged; so~is the
      left-hand-side term as soon as~$j>0$, so that the inequality is
      preserved for those cases. The case~$j=0$ is trivial.

    \item If $(\op,d)=(R,i)$ with $i>0$, then $q=s_i$ and~$q'=s_{i+1}$.
      we have $\state{\gamma'}=\state{\gamma}\Msetminus s_i \Mcup s_{i+1}$
      and $\data{\gamma'} = \data{\gamma}=i$.
      Again, along this read-transition, the right-hand side term is
      unchanged, while the left-hand-side term is unchanged for all $j\not=i$.

      It~remains to prove the inequality for~$j=i$. We~apply the induction
      hypothesis in~$\aconf$ for~$j=i-1$: since the transition $(q,R,i,q')$
      is available, it~must hold that $\state{\gamma}(s_i)\geq 1$ and
      ${\data{\gamma}=i=j+1}$. Hence 
      $\sum_{k=0}^{i-1} \state{\gamma}(s_k) \geq i-1+1=i$, and 
      $\sum_{k=0}^{i}\state{\gamma}(s_k) \geq {i+1}$.
      This implies $\sum_{k=0}^{i} \state{\gamma'}(s_k) \geq i$.
 
    \item If $(\op,d)=(W,i)$, then $q=s_i$ and $q'=s_0$. Thus
      $\state{\gamma'} = \state{\gamma}\Msetminus s_i\Mcup s_0$.
      For~$j=i-1$, the left-hand-side term of the inequality is increased
      by~$1$, while the right-hand-side one is either unchanged or also increased
      by~$1$. The property is preserved in both cases. 
      For~$j\not=i-1$, the left-hand-side term cannot decrease, while the
      right-hand-side term cannot increase. Hence 
      the invariant is preserved.
  \end{itemize}
  As a consequence, if $m<n$, we have (for $j=m$)
  $\sum_{k=0}^{m}\state{\gamma}(s_k) = m$ for any reachable
  configuration~$\aconf$, so that $\state{\gamma}(s_n) = 0$.

  \smallskip 
  Conversely, if $m\geq n$, from any
  configuration~$\aconf$, it~is possible to reach $\gamma_i =
  \tuple{s_0^{i}\Mcup s_{i+1}^{m-i},i}$ for any $0\leq i < n$:
  \begin{itemize}
    \item for~$i=0$: all processes can go to~$s_0$, then~write~$0$ in the
      register, and all move to~$s_1$:
      $\aconf\rightarrow^*\tuple{s_0^m,d}\xrightarrow{(W,0)}\tuple{s_0^m,0}
        {\xrightarrow{(R,0)}}{}^m\tuple{s_1^m,0}$;
    \item for $1<i<n-1$, assuming $\tuple{s_0^{i}\Mcup s_{i+1}^{m-i},i}$ can
      be reached, one of the processes in~$s_{i+1}$ can write~$i+1$ (going
      back to~$s_0$), and the remaining~$m-i-1$ processes in~$s_{i+1}$ can go
      to~$s_{i+2}$:
        \[\tuple{s_0^is_{i+1}^{m-i}, i}
            \xrightarrow{(W,i+1)}\tuple{s_0^{i+1}s_{i+1}^{m-i-1},i+1}
            \xrightarrow{(R,i+1)}{}^{m-i-1}\tuple{s_0^{i+1}s_{i+2}^{m-i-1},i+1}
        \]
  \end{itemize}
  Thus from any $\gamma\in\Gamma$, configuration~$\gamma_{n-1} =
  \tuple{s_0^{n-1}s_n^{m-n+1},n-1}$ is reachable. Furthermore,
  $\gamma_{n-1}$~contains the final state~$s_n$ since $m\geq n$.

  Hence, we deduce that there is a unique bottom strongly-connected component
  in~$\Sys_{\calF_n}^m$, and that $\gamma_{n-1}$ belongs to~it: this
  configuration is reached with probability $1$ from $\tuple{s_0^m,0}$.
  It~follows that $\pr(\tuple{s_0^m,0},\Objective{s_f}) = 1$.
\end{proof}

\subsection{Basic results}

In this section, we consider a register protocol
$\Prot = \tuple{Q, D, q_0, T}$, its associated distributed system
$\Sys_\Prot=\tuple{\Confs,\Prob}$, an initial register value $d_0\in D$ and a
target state $q_f \in Q$. We define a partial order $\preceq$ over the set~$\Confs$ of configurations as
follows: $\tuple{\mu,d} \preceq \tuple{\mu',d'}$ if, and only~if,  $d=d'$ and
$\support[']{\mu}=\support{\mu'}$ and $\mu\Msubseteq \mu'$. Note that with respect to the classical order
over multisets, we require here that the supports of~$\mu$ and~$\mu'$ be
the same (we add in fact a finite information to hold for the comparison). We~know from Dickson's
lemma that $\tuple{\bbN^Q,\Msubseteq}$ is a wqo and since $Q$, $D$ and
the supports of multisets in $\bbN^Q$ are
finite, we can deduce the following lemma.

\begin{lemma}\label{lem:conf-wqo}
$\tuple{\Confs,\preceq}$ is a wqo.
\end{lemma}

We will give some properties of register protocols, but
first we introduce some further notations. Given a set of
configuration $\SubConfs \subseteq \Confs$, we~define
$\Pre^\ast(\SubConfs)$ and $\Post^\ast(\SubConfs)$ as follows:
\begin{xalignat*}2
\Pre^\ast(\SubConfs)&=\set{\aconf \in \Confs \mid \exists
    \aconf' \in \SubConfs. \aconf \rightarrow^\ast \aconf'}
&
\Post^\ast(\SubConfs)&=\set{\aconf' \in \Confs \mid \exists
    \aconf \in \SubConfs. \aconf \rightarrow^\ast \aconf'}
\end{xalignat*}

We also define the set~$\ConfsF$ of configurations we aim to
reach as $\set{\aconf \in \Confs \mid
  \state{\aconf}(q_f)>0}$. It~holds that $\aconf \in
\Pre^\ast(\ConfsF)$ if, and only~if, there exists a path in $\Objective{q_f}$
starting in~$\aconf$.  

\looseness=-1
As already mentioned, when $\tuple{\mu,d} \rightarrow
\tuple{\mu',d'}$ in $\Sys_\Prot$, 
the
multisets $\mu$ and~$\mu'$ have the same cardinality. This~implies that given
$k >0$, the set 
$\Post^\ast(\set{\tuple{q^k_0,d_0}})$ is finite (remember that $Q$ and $D$ are
finite). As~a~consequence, for a fixed~$k$, checking
whether $\pr(\tuple{q^k_0,d_0},\Objective{q_f})=1$ can be easily
achieved by analyzing the finite-state discrete-time Markov
chain~$\Sys_\Prot^k$~\cite{PoMC2008-BK}.

\begin{lemma}
\label{lem:prob-pre-post}
Let $k\geq 1$. Then $\pr(\tuple{q^k_0,d_0},\Objective{q_f})=1$
 if, and only~if,
$\Post^\ast(\set{\tuple{q_0^k,d_0}}) \subseteq \Pre^\ast(\ConfsF)$. 
\end{lemma}

\looseness=-1
The difficulty here precisely lies in
finding such a~$k$  and in proving that, once we have
  found one correct value for $k$, all larger values are correct as
  well (to get the cut-off property).
Characteristics of register protocols provide us with some tools to solve this
problem. We base our analysis on reasoning on the set of
configurations reachable from initial configurations in
$\upcl\set{\tuple{q_0,d_0}}$  (the~upward closure
of $\set{\tuple{q_0,d_0}}$ w.r.t.~$\preceq$), remember that since the order
$\tuple{\Confs,\preceq}$ requires equality of support for elements to
be comparable, we have that $\upcl\set{\tuple{q_0,d_0}}= \bigcup_{k\geq 1}
\set{\tuple{q_0^k,d_0}}$. We~begin by showing that this set of
reachable configurations and the set of configurations
from which $\ConfsF$ is reachable are both upward-closed. Thanks to
Lemma~\ref{lem:conf-wqo}, they can be represented as upward closures of finite
sets. To~show that $\Post^\ast(\upcl\set{\tuple{q_0,d_0}})$ is upward-closed, we prove
that register protocols enjoy the following monotonicity property. A~similar
property is given in~\cite{DEGM15} and
derives from the non-atomicity of operations.

\begin{lemma}
\label{lem:copycat}
Let $\aconf_1$, $\aconf_2$, and $\aconf'_2$ be configurations in~$\Confs$. If
$\aconf_1 
\rightarrow^\ast \aconf_2$ and $\aconf_2 \preceq \aconf'_2$, then there
exists $\aconf'_1 \in \Confs$ such that $\aconf'_1
\rightarrow^\ast \aconf'_2$ and $\aconf_1 \preceq \aconf'_1$.
\end{lemma}

\begin{proof}
Assume $\gamma_1\rightarrow\gamma_2$ with transition $(q_1,\op,d,q_2)\in T$
and $\gamma_2\preceq \gamma'_2$. Let
$k=\state{\gamma'_2}(q_2)-\state{\gamma_2}(q_2)\geq 0$. Then
$\state{\gamma'_2}(q_2) = k+\state{\gamma_2}(q_2) \geq k+1$ so we define
$\gamma'_1=\tuple{ \state{\gamma'_2}\Msetminus q_2^{k+1}\Mcup q_1^{k+1},
  \data{\gamma_1} }$. Then:
\begin{itemize}
  \item if $\op=W$, the path $\gamma'_1\rightarrow^\ast\gamma'_2$, obtained by
    performing $k+1$ times the transition $(q_1,W,d,q_2)$, is a valid
    path since 
    write operations can always be performed, independently of the content of
    the register;
  \item if $\op=R$, the path $\gamma'_1\rightarrow^\ast\gamma'_2$, defined by
    applying $k+1$ times the transition $(q_1,R,d,q_2)$, is also a
    valid path, since the data~$d$ in the register is unchanged.
\end{itemize}
By construction, we have $\state{\gamma_1}(q_2)=\state{\gamma'_1}(q_2)$,
and $1\leq \state{\gamma_1}(q_1)\leq\state{\gamma'_1}(q_1)$, 
and $\state{\gamma_1}(q)=\state{\gamma'_1}(q)$ for all $q\neq q_2$. Hence
$\gamma_1\preceq \gamma'_1$.
The result is then generalized to arbitrary path
$\gamma_1\rightarrow^\ast\gamma_2$ by induction.
\end{proof}

$\Pre^\ast(\ConfsF)$ is also upward-closed, since if $\ConfsF$ can
be reached from some configuration~$\aconf$, it~can also be reached
by a larger configuration by keeping the extra copies~idle.
Thus:
\begin{lemma}\label{lem:prepost-upwardclosed}
$\Post^\ast(\upcl\set{\tuple{q_0,d_0}})$ and $\Pre^\ast(\ConfsF)$ are
upward-closed sets in~$\tuple{\Confs,\preceq}$.
\end{lemma}

\subsection{Existence of a cut-off}

From Lemma~\ref{lem:prepost-upwardclosed}, and from the fact that
$\tuple{\Confs,\preceq}$ is a wqo, 
there must exist two finite
sequences of configurations $(\theta_i)_{1 \leq i \leq n}$ and $(\eta_i)_{1
  \leq i \leq m}$ such that ${\Post^\ast(\upcl\set{\tuple{q_0,d_0}})=
\upcl\set{\theta_1,\ldots,\theta_n}}$ and
$\Pre^\ast(\ConfsF)=\upcl\set{\eta_1,\ldots,\eta_m}$. 
By~analyzing these two sequences, we now prove that any register protocol
has a cut-off (for~any initial register value and any target
location).

We let $\Delta, \Delta' \subseteq \Confs$ be two
upward-closed sets (for $\preceq$). We say that \emph{$\Delta$ is
  included in $\Delta'$ modulo single-state incrementation} whenever for
every $\gamma \in \Delta$, for every $q \in
\support{\aconf}$, there is some $k \in \bbN$ such that
$\aconf \Mcup q^k \in \Delta'$. Note that this condition can be
checked using only
comparisons between minimal elements of $\Delta$ and $\Delta'$. In
particular, we have the following lemma.
\begin{lemma}
  \label{lemma:inc}
  $\Post^\ast(\upcl\set{\tuple{q_0,d_0}})$ is included in
  $\Pre^\ast(\ConfsF)$ modulo single-state incrementation if, and only~if, for
  all $i \in [1;n]$, and 
  for all~$q\in\support{\theta_i}$,
  there exists $j \in [1;m]$ such that
  $\data{\theta_i}=\data{\eta_j}$ and $\support{\theta_i} =
  \support[\theta]{\eta_j}$ and $\state{\eta_j}(q') \leq
  \state{\theta_i}(q')$ for all $q' \in Q \setminus\set{q}$.
\end{lemma}
\begin{proof}
Suppose that $\Post^\ast(\upcl\set{\tuple{q_0,d_0}})$ is included in
  $\Pre^\ast(\ConfsF)$ modulo single-state incrementation. Let $i \in [1;n]$ and $q \in
  \support{\theta_i}$. By definition, there exists some $k \in \bbN$
  such that  $\theta_i \Mcup q^k \in \Pre^\ast(\ConfsF)$. Hence there is
  $j \in [1;m]$ such that $\eta_j \preceq \theta_i \Mcup q^k$. Hence
  we have 
  $\data{\theta_i}=\data{\eta_j}$, $\support{\theta_i} =
  \support{\eta_j}$ and $\state{\eta_j}(q') \leq
  \state{\theta_i}(q')$ for all $q' \in Q \setminus\set{q}$.

Now assume that  for
  all $i \in [1;n]$, and 
  for all~$q\in\support{\theta_i}$,
  there exists $j \in [1;m]$ such that
  $\data{\theta_i}=\data{\eta_j}$, $\support{\theta_i} =
  \support[\theta]{\eta_j}$ and $\state{\eta_j}(q') \leq
  \state{\theta_i}(q')$ for all $q' \in Q \setminus\set{q}$. Let
  $\aconf \in \Post^\ast(\upcl\set{\tuple{q_0,d_0}})$. Hence there
  exists $i \in [1;n]$ such that $\theta_i \preceq \aconf$ (note that
  hence $\support{\theta_i}=\support{\aconf}$). Let $q \in
  \support{\aconf}$. Then there exists $j \in [1;m]$ such that
  $\data{\theta_i}=\data{\eta_j}$, $\support{\theta_i} =
  \support[\theta]{\eta_j}$ and $\state{\eta_j}(q') \leq
  \state{\theta_i}(q')$ for all $q' \in Q \setminus\set{q}$. Take
  $k=|\state{\eta_j}(q)-\state{\theta_i}(q)|$. We consider the configuration $\aconf'=\aconf
  \Mcup q^k$. For all $q' \in Q \setminus\set{q}$, we have $\state{\eta_j}(q') \leq
  \state{\theta_i}(q') \leq \state{\aconf'}(q')$. And we have
  $\state{\eta_j}(q) \leq \state{\theta_i}(q) + k \leq
  \state{\aconf'}(q)$. This allows us to deduce that $\eta_j \preceq
  \aconf \Mcup q^k$ and consequently  $\aconf \Mcup q^k \in
  \Pre^\ast(\ConfsF)$. Consequently $\Post^\ast(\upcl\set{\tuple{q_0,d_0}})$ is included in
  $\Pre^\ast(\ConfsF)$ modulo single-state incrementation.
\end{proof}

Using the previous characterization of inclusion modulo single-state
incrementation for $\Post^\ast(\upcl\set{\tuple{q_0,d_0}})$ and
$\Pre^\ast(\ConfsF)$ together with the result of Lemma
\ref{lem:prob-pre-post}, we~are able to provide a first characterization of
the existence of a negative cut-off.

\begin{lemma}
\label{lem:negative-cut-off}
 If $\Post^\ast(\upcl\set{\tuple{q_0,d_0}})$ is not
included in $\Pre^\ast(\ConfsF)$ modulo single-state incrementation, then
$\max_{1 \leq i \leq n}(\Msize{\state{\theta_i}} )$ is a
negative cut-off.
\end{lemma}

\begin{proof}
  Applying the previous lemma, there is $i \in [1;n]$ and
  $q\in\support{\theta_i}$ such that for every $j\in [1;m]$, either
  $\data{\theta_i}\neq \data{\eta_j}$, or $\support{\theta_i} \neq
  \support{\eta_j}$, or there is $q_j\neq q$ such that
  $\state{\eta_j}(q_j) > \state{\theta_i}(q_j)$.

  Let $k_i = \Msize{\state{\theta_{i}}}$, and fix $k \ge k_i$.  We
  define $\gamma_{i,k}=\theta_i \Mcup q^{k-k_i}$.  Clearly $\theta_i
  \preceq \gamma_{i,k} \in \Post^\ast(\{q_0^k,d_0\})$. On the
  opposite, for every $j \in [1;m]$, $\eta_j \not\preceq
  \gamma_{i,k}$; hence we conclude that $\gamma_{i,k}\not\in
  \Pre^\ast(\ConfsF)$.

  Applying Lemma~\ref{lem:prob-pre-post}, we get that
  $\pr(\tuple{q^k_0,d_0},\Objective{q_f})<1$ for every $k \ge k_i$.
\end{proof}

We now prove that if the condition of Lemma~\ref{lem:negative-cut-off} fails
to hold, then there is a positive cut-off.In~order to make our claim precise,
for every $i \in [1;n]$ and for any $q \in \support{\theta_i}$, we~let
$d_{i,q}=\max\{(|\state{\eta_j}(q) - \state{\theta_i}(q)|) \mid 1 \leq
j \leq m\ \text{and}\ \support{\theta_i} = \support{\eta_j}\}$.

\begin{lemma}
\label{lem:positive-cut-off}
If $\Post^\ast(\upcl\set{\tuple{q_0,d_0}})$ is included in
$\Pre^\ast(\ConfsF)$ modulo single-state incrementation, then $\max_{1
  \leq i \leq n}(\Msize{\state{\theta_i}} + \sum_{q\in
  \support{\theta_i}} d_{i,q} )$ is a  positive cut-off.
\end{lemma}

\begin{proof}
  Let $k_0 = \max_{1 \leq i \leq n}(\Msize{\state{\theta_i}} +
  \sum_{q\in \support{\theta_i}} d_{i,q} )$, and $k \ge
  k_0$. Consider a configuration $\gamma\in\Post^\ast(\tuple{q_0^k,d_0})$.  There
  exists $i\in[1;n]$ with $\theta_i \preceq \gamma$.

  Choose $q \in \support{\theta_i} = \support{\gamma}$ such that
  $\state{\gamma}(q) \ge \state{\theta_i}(q)+d_{i,q}$ (this $q$ should
  exist since $\Msize{\state{\gamma}} \ge k_0$).
  We apply Lemma~\ref{lemma:inc}, and exhibit $j \in [1;m]$ such that
  $\data{\theta_i} = \data{\eta_j}$, $\support{\theta_i} =
  \support{\eta_j}$ and for every $q' \ne q$, $\state{\eta_j}(q') \le
  \state{\theta_i}(q')$. 
  Now, $\state{\eta_j}(q) \le \state{\theta_i}(q) + d_{i,q} \le
  \state{\gamma}(q)$. We conclude that $\eta_j \preceq \gamma$, and
  therefore that $\Post^\ast(\set{\tuple{q_0^k,d_0}}) \subseteq
  \Pre^\ast(\ConfsF)$.  By Lemma~\ref{lem:prob-pre-post}, we conclude
  that $k_0$ is a positive cut-off.
\end{proof}

The last two lemmas entail our first result:
\begin{theorem}
\label{thm-cutoff}
Any register protocol admits a cut-off (for any given initial register value and
target state).
\end{theorem}

\section{Detecting negative cut-offs}
\label{sec-algo}

We develop an algorithm for deciding whether a distributed system
associated with a register protocol has a negative cut-off. Thanks to
Theorem~\ref{thm-cutoff}, this can also be used to detect the
existence of a positive cut-off. Our~algorithm relies on the
construction and study of a \emph{symbolic graph}, 
as we define below: for any given protocol~$\Prot$, the~symbolic graph
has bounded size, but can be used to reason about \emph{arbitrarily large}
distributed systems built from~$\Prot$. It~will store sufficient
information to decide the existence of a negative cut-off.

\subsection{\(k\)-bounded symbolic graph}

In this section, we consider a register protocol
$\Prot=\tuple{Q,D,q_0,T}$,  its associated distributed system
$\Sys_\Prot=\tuple{\Confs,\Prob}$, an initial register value $d_0 \in
D$, and  a target location~$q_f \in Q$ of~$\Prot$.
With~$\Prot$, we~associate a finite-state graph, called
\emph{symbolic graph of index~$k$}, which for $k$ large enough 
contains enough information
to decide 
the existence of a negative cut-off.

\begin{definition}
Let~$k$ be an integer.
The~\emph{symbolic graph of
  index~$k$} associated with~$\Prot$ and~$d_0$ is the transition
system~$\SG=\tuple{V,v_0,E}$ where
\begin{itemize}
\item $V = \bbN^Q_k \times 2^Q\times D$ contains triples made of a multiset of
  states of~$Q$ of size~$k$, a~subset of~$Q$, and the content of 
  the register; 
  the~multiset (called \emph{concrete part}) is used to
  exactly keep track of a fixed set of $k$ processes, while the subset of $Q$
  (the~\emph{abstract part}) encodes the support of the arbitrarily many remaining
  processes;
\item $v_0=\tuple{q_0^k, \{q_0\}, \{d_0\}}$;
\item transitions are of two types, depending whether they involve a process
  in the concrete part or a process in the abstract part. Formally, there is a
  transition $\tuple{\mu,S,d} \to \tuple{\mu',S',d'}$
  whenever there is a transition $(q,O,d'',q')\in T$ such that $d=d'=d''$ if
  $O=R$ and $d'=d''$ if~$O=W$, 
  and one of the following two conditions holds:
  \begin{itemize}
  \item either $S'=S$ and $q\Msubseteq\mu$ (that is, $\mu(q)>0$) and
    $\mu'=\mu\Msetminus q\Mcup q'$;
  \item or $\mu=\mu'$ and $q\in S$ and $S'\in\{S\setminus\{q\}\cup
    \{q'\},S\cup \{q'\}\}$.
  \end{itemize}
\end{itemize}
\end{definition}

The symbolic graph of index~$k$ can be used as an abstraction of distributed
systems made of at~least $k+1$ copies of~$\Prot$:  
it~keeps full information of
the states of $k$ processes, and only gives the support of the states
of the other processes. In~particular, the~symbolic graph of index~$0$
provides only the states appearing in each configuration of the system.
\begin{figure}[htb]
\centering
\begin{tikzpicture}[xscale=1.1,yscale=.9]
\begin{scope}
\tikzstyle{sg}=[draw,rounded corners=2mm,minimum height=6mm,inner sep=3pt]
\draw (0,0) node[sg,fill=fbleuc] (q0-0) {$\{q_0\},0$};

\draw (2,-.5) node[sg,fill=fvert] (q1-1) {$\{q_1\},1$};
\draw (2,-1.5) node[sg,fill=fvert] (q1-0) {$\{q_1\},0$};
\draw (8,-.5) node[sg,fill=fvert] (q1-2) {$\{q_1\},2$};

\draw (5,-.5) node[sg,fill=fjaune] (q2-1) {$\{q_2\},1$};

\draw (2,1.5) node[sg,path picture={\foreach \i in {-2,-1.5,...,2}
  {\draw[line width=2.5mm,fvert] (\i+.25,-1) -- +(0,2);
   \draw[line width=2.5mm,fbleuc] (\i,-1) -- +(0,2);}}] (q01-0) {$\{q_0,q_1\},0$};
\draw (2,.5) node[sg,path picture={\foreach \i in {-2,-1.5,...,2}
  {\draw[line width=2.5mm,fvert] (\i+.25,-1) -- +(0,2);
   \draw[line width=2.5mm,fbleuc] (\i,-1) -- +(0,2);}}] (q01-1) {$\{q_0,q_1\},1$};
\draw (8,1.5) node[sg,path picture={\foreach \i in {-2,-1.5,...,2}
  {\draw[line width=2.5mm,fvert] (\i+.25,-1) -- +(0,2);
   \draw[line width=2.5mm,fbleuc] (\i,-1) -- +(0,2);}}] (q01-2) {$\{q_0,q_1\},2$};

\draw (5,1.5) node[sg,path picture={\foreach \i in {-2,-1.5,...,2}
  {\draw[line width=2.5mm,fjaune] (\i+.25,-1) -- +(0,2);
   \draw[line width=2.5mm,fbleuc] (\i,-1) -- +(0,2);}}] (q02-1) {$\{q_0,q_2\},1$};

\draw (5,.5) node[sg,path picture={\foreach \i in {-2,-1.25,...,2}
  {\draw[line width=2.5mm,fjaune] (\i+.25,-1) -- +(0,2);
   \draw[line width=2.5mm,fbleuc] (\i+.5,-1) -- +(0,2);
   \draw[line width=2.5mm,fvert] (\i,-1) -- +(0,2);}}] (q012-1) {$\{q_0,q_1,q_2\},1$};
\draw (8,.5) node[sg,path picture={\foreach \i in {-2,-1.25,...,2}
  {\draw[line width=2.5mm,fjaune] (\i+.25,-1) -- +(0,2);
   \draw[line width=2.5mm,fbleuc] (\i+.5,-1) -- +(0,2);
   \draw[line width=2.5mm,fvert] (\i,-1) -- +(0,2);}}] (q012-2) {$\{q_0,q_1,q_2\},2$};

\draw (5,-1.5) node[sg,path picture={\foreach \i in {-2,-1.5,...,2}
  {\draw[line width=2.5mm,fjaune] (\i+.25,-1) -- +(0,2);
   \draw[line width=2.5mm,fvert] (\i,-1) -- +(0,2);}}] (q12-1) {$\{q_1,q_2\},1$};
\draw (8,-1.5) node[sg,path picture={\foreach \i in {-2,-1.5,...,2}
  {\draw[line width=2.5mm,fjaune] (\i+.25,-1) -- +(0,2);
   \draw[line width=2.5mm,fvert] (\i,-1) -- +(0,2);}}] (q12-2) {$\{q_1,q_2\},2$};

\draw (11,0) node[sg,fill=frouge,minimum height=3.3cm,text width=1.75cm,align=center] (qf)
      {all sets\\ containing $q_f$};
\draw[-latex',rounded corners=2mm] (q0-0) |- (q1-0);
\draw[-latex',rounded corners=2mm] (q0-0) |- (q01-0);
\draw (q01-0) edge[-latex',out=210,in=150] (q1-0);
\draw[-latex'] (q01-0) -- (q01-1);
\draw[-latex'] (q1-0) -- (q1-1);
\draw[-latex'] (q01-1) -- (q012-1);
\draw[-latex'] (q01-1) -- (q02-1);
\draw[-latex'] (q1-1) -- (q12-1);
\draw[-latex'] (q1-1) -- (q2-1);
\draw[-latex'] (q02-1) -- (q012-2);
\draw[-latex'] (q02-1) -- (q01-2);
\draw (q01-2) edge[-latex',out=160,in=40] (q01-1.30);
\draw[-latex'] (q012-1) -- (q01-2);
\draw[-latex'] (q012-1) -- (q02-1);
\draw[latex'-latex'] (q012-1) -- (q012-2);
\draw[-latex'] (q2-1) -- (q12-2);
\draw[-latex'] (q12-1) -- (q2-1);
\draw[-latex'] (q12-1) -- (q2-1);
\draw[-latex'] (q2-1) -- (q1-2);
\draw[-latex'] (q12-1) -- (q1-2);
\draw (q1-2) edge[-latex',out=165,in=15] (q1-1);
\draw[latex'-latex'] (q12-1) -- (q12-2);
\draw[-latex'] (q12-2) -- (q1-2);
\draw[-latex'] (q012-2) -- (q01-2);
\draw[-latex'] (q12-2) -- ($(qf.180)+(0,-1.5)$);
\draw[-latex'] (q012-2) -- ($(qf.180)+(0,.5)$);
\end{scope}
\end{tikzpicture}
\caption{Symbolic graph (of index~$0$) of the protocol of
  Fig.~\ref{fig-ex-protocol} (self-loops omitted).}\label{fig-ex-symbg}
\end{figure}
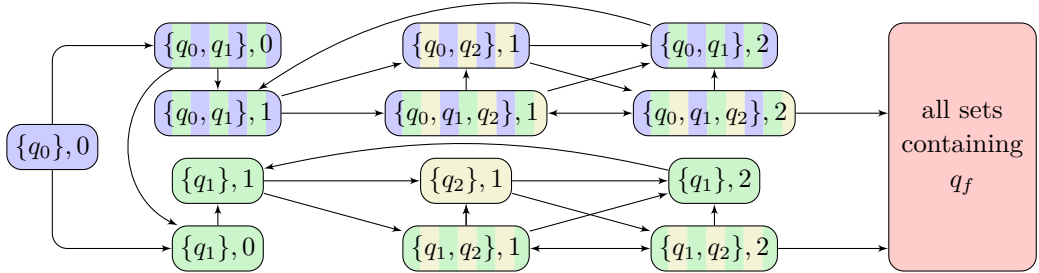

\begin{subexample}{running}\label{ex-1c}
Consider the protocol depicted in Fig.~\ref{fig-ex-protocol}. Its symbolic
graph of index~$0$ is depicted in Fig.~\ref{fig-ex-symbg}.
Notice that the final state (representing all configurations containing~$q_f$)
is reachable from any state of this symbolic
graph. However, our original protocol~$\Prot$ of Fig.~\ref{fig-ex-protocol}
does not have a positive cut-off (assuming initial register value~$0$):
indeed, with positive probability, a~single process will go to~$q_1$
and immediately
write~$1$ in the register, thus preventing any other process to leave~$q_0$;
then one may check that the process in~$q_1$ alone cannot reach~$q_f$, so that
the probability of reaching~$q_f$ from~$q_0^k$ is strictly less than~$1$, for
any~$k>0$. 
This livelock is not taken into account in the symbolic graph of~index~$0$,
because from any configuration with support~$\{q_0,q_1\}$ and register
data equal to~$1$, the~symbolic graph
has a transition to the configuration with support~$\{q_0,q_1,q_2\}$, which
only exists in the concrete system when there are at least two processes
in~$q_1$. As we prove in the following, analyzing the symbolic graph for a
sufficiently large index guarantees to detect such a situation. 
\end{subexample}

\enlargethispage{2mm}
For any index~$k$, the symbolic graph achieves the following correspondence: 
\begin{lemma}
\label{lemma-gsymb}
Given two states~$\tuple{\mu,S,d}$ and~$\tuple{\mu',S',d'}$,
there is a 
transition from $\tuple{\mu,S,d}$ 
to~$\tuple{\mu',S',d'}$ in the symbolic graph \SG of index~$k$ if, and only~if,
there exist multisets~$\delta$ and $\delta'$ with respective supports~$S$
and~$S'$, and such that $\tuple{\mu\Mcup\delta,d}\to
\tuple{\mu'\Mcup\delta',d'}$ in $\Sys_\Prot$. 
\end{lemma}

\begin{proof}
We begin with the reverse implication: if there is a
transition from $\tuple{\mu\Mcup\delta,d}$ to
$\tuple{\mu'\Mcup\delta',d'}$ 
(assuming it~is a write transition, the other case being similar) in the
distributed system, then this transition originates from a
transition~$(q,W,d',q')$ in~$\Prot$, and either this transition affects a
process from the set of~$k$ processes that are monitored exactly by the symbolic
graph, or it affects a process in the abstract part, in which only the
support is monitored. In~the former case, $\delta=\delta'$, hence also their
supports are equal, and the transition~$(q,W,d',q')$ is applied to a location
in~$\mu$, which entails $q\Msubseteq \mu$ and $\mu'=\mu\Msetminus
q\Mcup q'$ and $d'=d''$. In~the latter case, we~get $\mu=\mu'$, and the
transition~$(q,W,d',q')$ is applied to a state in the support, so that $q\in
S$ and $S'$~is either $S\cup\{q'\}$ (in~case $\delta(q)>1$), or
  it~is~$S\setminus\{q\}\cup\{q'\}$ (in~case $\delta(q)=1$).

Conversely, if there is a transition $\tuple{\mu,S,d} \to
\tuple{\mu',S',d'}$ 
(assuming it originates from a $W$-transition $(q,W,d',q')$ in~$\Prot$, the
other case being similar), we~again have to consider two separate cases.
\begin{itemize}
\item The first case is when $S'=S$, $q\Msubseteq\mu$ and
  $\mu'=\mu\Msetminus q\Mcup q'$, corresponding to the case where the
  transition is performed by one of the $k$ processes tracked exactly by the
  symbolic graph. In~that case, for any~$\delta$ with support~$S$, there is a
  transition from~$\tuple{\mu\Mcup\delta,d}$
  to~$\tuple{\mu'\Mcup\delta,d'}$ in the 
  concrete distributed system. 
\item In the second case, $\mu'=\mu$, $q\in S$, and $S'$ is either
  $S\setminus\{q\}\cup\{q'\}$ or $S\cup \{q'\}$. Consider any
  multiset~$\delta$ with support~$S$, and such that $\delta(q)>1$ in case
  $S'=S\cup\{q'\}$, and $\delta(q)=1$ if $S'=S\setminus \{q\}\cup\{q'\}$. 
  Let~$\delta'=\delta\Msetminus q\Mcup q'$; then the support of~$\delta'$
  is~$S'$, and there is a transition from~$\tuple{\mu\Mcup\delta,d}$
  to~$\tuple{\mu'\Mcup \delta', d'}$, as required.
\end{itemize}
This concludes our proof.
\end{proof}

\subsection{Deciding the existence of a negative cut-off}

We now explain how the symbolic graph can be used to decide the existence of a
negative cut-off. 
Since
$\Pre^\ast(\ConfsF)$ is upward-closed in~$\tuple{\Confs,\preceq}$,
there is a finite set of configurations $\{\eta_i=\tuple{\mu_i,d_i} \mid
1\leq i\leq m\}$ such that $\Pre^\ast(\ConfsF)=\upcl\set{\eta_i\mid
  1\leq i\leq m}$. We~let $K=\max \{ \state{\eta_i}(q) \mid q\in Q,\
1\leq i\leq m\}$, and show that for our purpose, it~is
enough to consider the symbolic graph of index $K\cdot|Q|$; 
we provide a bound on~$K$ in the next section.

\begin{lemma}
\label{lem:algo}
There is a negative cut-off for~$\Prot$, $d_0$ and~$q_f$
if, and only~if, there is a node in the symbolic graph
of index~$K\cdot |Q|$ that is reachable from~$\tuple{q_0^{K\cdot|Q|},\{q_0\}, d_0}$ 
but from which no configuration involving~$q_f$ is reachable.
\end{lemma}

\begin{proof}
  We~begin with the converse implication, assuming that there is a
  state~$\tuple{\mu,S,d}$ in the symbolic graph of index~$K\cdot |Q|$
  that is reachable from~$(q_0^{K\cdot |Q|},\{q_0\},d_0)$ and from
  which no configuration in~$\ConfsF$ is reachable. Applying
  Lemma~\ref{lemma-gsymb}, there exist multisets~$\delta_0=q_0^N$
  and~$\delta$, with respective supports $\{q_0\}$ and~$S$, such that
  $\tuple{\mu\Mcup\delta,d}$ is reachable from~${\tuple{q_0^{K\cdot
        |Q|}\Mcup\delta_0,d_0}}$. If~location~$q_f$ was reachable
  from~$\tuple{\mu\Mcup\delta,d}$ in the distributed system, then
  there would exist a path from~$\tuple{\mu,S,d}$ to a state
  involving~$q_f$ in the symbolic graph, which contradicts our
  hypothesis. By Lemma~\ref{lem:copycat}, it follows that such a configuration $\tuple{\mu\Mcup\delta',d}$\,---\,which cannot reach $q_f$\,---\,can be reached from ${\tuple{q_0^{K\cdot
        |Q|}\Mcup q_0^{N'},d_0}}$ for any $N' \geq N$: hence it~cannot be the
  case that $q_f$ is reachable almost-surely for any $N' \geq N$. Therefore there cannot be a
  positive cut-off, which implies that there is a negative one (from
  Theorem~\ref{thm-cutoff}).

Conversely, 
if there is a negative cut-off, then for
some~$N>K\cdot |Q|$, the~distributed system~$\Sys_\Prot^N$ with~$N$
processes  has probability less than~$1$ of reaching~$\ConfsF$ from~$q_0^N$. This~system
being finite, 
there must exist a reachable configuration~$\tuple{\mu,d}$ from which $q_f$ is
not reachable~\cite{PoMC2008-BK}. Hence
$\tuple{\mu,d}\notin\Pre^\ast(\ConfsF)$, 
and
for all
$i\leq m$, there~is a location~$q^i$ such that $\mu(q^i)<\mu_i(q^i)\leq K$.
Then there must exist a reachable state $\tuple{\kappa,S,d}$ of the
symbolic 
graph of index~$K\cdot |Q|$ for which $\kappa(q^i)=\mu(q^i)$ and
$q^i\notin S$, for all $1\leq i\leq m$: it~indeed suffices to follow the path
from~$\tuple{q_0^N,d_0}$ to~$\tuple{\mu,d}$ while keeping track of the
processes 
that end up in some~$q^i$ in the concrete part; this~is possible because the
concrete part has size at least~$K\cdot |Q|$. 

It~remains to be proved that no state involving~$q_f$ is reachable
from~$\tuple{\kappa,S,d}$ in the symbolic graph. If~it were the case, then by
Lemma~\ref{lemma-gsymb}, there would exist~$\delta$ with support~$S$ such that
$\ConfsF$ is reachable from~$\tuple{\kappa\Mcup\delta,d}$ in the distributed
system. Then $\tuple{\kappa\Mcup\delta,d}\in \Pre^\ast(\ConfsF)$, so~that for
some~$1\leq i\leq m$, $(\kappa\Mcup\delta)(q^i)\geq \mu_i(q^i)$, which is not
possible as $\kappa(q^i)<\mu_i(q^i)$ and $q^i$~is not in the support~$S$
of~$\delta$. This contradiction concludes the proof.
\end{proof}

\begin{remark}
Besides the existence of a negative cut-off, this proof also provides us with
an upper bound on the tight cut-off,  as we shall see in Section~\ref{sec-bounds}.
\end{remark}

\subsection{Complexity of the algorithm}

We now consider the complexity of the algorithm that can be deduced from Lemma~\ref{lem:algo}.
Using results by Rackoff on the coverability problem in Vector Addition
Systems~\cite{Rac78}, we can bound~$K$\,--\,and
consequently the size of the needed symbolic graph\,--\,by a
\textit{double-exponential} in the size of the protocol. Therefore, it
suffices to solve a reachability problem in
\NLOGSPACE~\cite{Sip97} on this doubly-exponential graph:
this boils down to $\NEXPSPACE$ with regard to the protocol's size, hence
$\EXPSPACE$ by Savitch's theorem~\cite{Sip97}.

\begin{theorem}
  \label{thm:expspace-membership}
Deciding the existence of a negative cut-off is in \EXPSPACE.
\end{theorem}

\begin{proof}
  Recall that $\Pre^\ast(\ConfsF)$ is exactly the set of configurations that
  can cover $q_f$, i.e., configurations $\aconf$ from which there exists a
  path $\aconf \rightarrow^\ast \aconf'$ with $\state{\aconf'}(q_f) > 0$.
  Recall also that it can be written as an upward-closure of minimal elements:
  $\Pre^\ast(\ConfsF) = \upcl\set{\eta_1,\ldots,\eta_m}$. Now,
  consider the value $K$ in Lemma~\ref{lem:algo}: it is defined as $K=\max \{
  \state{\eta_i}(q) \mid q\in Q,\ 1\leq i\leq m\}$, i.e., the maximum number
  of states appearing in any multiset of any minimal configuration~$\eta_i$.
  The value of $K$ can be bounded using classical results on the coverability
  problem in Vector Addition Systems (VAS)~\cite{Rac78}.

  Intuitively, a $b$-dimensional VAS is a system composed of an initial
  $b$-dimensional vector~$\mathbf{v_{0}}$ of naturals (the~\textit{axiom}),
  and a finite set of $b$-dimensional integer vectors (the~\textit{rules}).
  An~\textit{execution} is built as follows: it starts from the axiom and, at
  each step, the next vector is derived from the current one by adding a rule,
  provided that this derivation is \textit{admissible}, i.e., that the
  resulting vector only contains non-negative integers. An execution ends if
  no derivation is admissible. The \textit{coverability problem} asks if a
  given target vector $\mathbf{v} = (v_1, \ldots{}, v_b)$ can be covered,
  i.e., does there exists a (possibly extendable) execution $\mathbf{v_{0}}
  \rightsquigarrow \mathbf{v_{1}} \rightsquigarrow \ldots{} \rightsquigarrow
  \mathbf{v_{n}} = \mathbf{v'}$ such that, for all $1\leq i \leq b$, it~holds that $v_i \leq v'_i$.

  Our distributed system $\Sys_\Prot$ can be seen as a $\Msize Q$-dimensional
  VAS where each transition is modeled by a rule vector modifying the multiset
  of the current configuration. Formally, one has to take into account that
  available rules depend on the data stored in the shared register. This
  can be achieved by either considering the expressively equivalent model of
  VAS with states (VASS, see~e.g.,~\cite{RY86}) or by
  adding $\mathcal{O}(\Msize D)$ dimensions to enforce this restriction. Over such a VAS(S), we
  are interested in the coverability of the vector corresponding to the
  multiset~$q_f$ (i.e.,~containing only one copy of~$q_f$ and no other state).
  In~particular, we want to bound the size of vectors needed to cover~$q_f$,
  as it will lead to a bound on minimal elements~$\eta_i$ of
  $\Pre^\ast(\ConfsF)$, hence a bound on the value~$K$.

  Results by Rackoff (hereby as reformulated by Demri
  \textit{et~al.}~\cite[Lemma 3]{DJLL13}) state
  that if a covering execution exists from an initial
  vector~$\mathbf{v_0}$, then there is one whose length may be
  doubly-exponential in the size of the input: singly-exponential in
  the size of the rule set and the target vector, and
  doubly-exponential in the dimension of the~VAS. Hence, for our
  distributed system $\Sys_\Prot$, seen as a VAS, this implies that if
  $q_f$ can be covered from a configuration $\aconf$, there is a
  covering execution whose length is bounded by
  some~$L$ in $2^{\calO\left(\Msize
        Q \cdot \Msize D\right)^{\mathcal{O}\left(\Msize
        Q + \Msize D\right)}}$.
  This bound on the \textit{length} of the execution
  obviously also implies a bound on the \textit{number of processes}
  actively involved in the execution (because at each transition, only
  one process is active). Hence, we can deduce that if a configuration
  $\aconf = \tuple{\mu,d}$ can cover $q_f$ (i.e., there exists a path
  $\aconf \rightarrow^\ast \aconf'$ with $\state{\aconf'}(q_f) > 0$),
  then it is also the case of configuration $\aconf'' =
  \tuple{\mu'',d}$, which we build as follows: $\forall\, q \in Q,\;
  \mu''(q) = \min\{\mu(q), L\}$. That is, it also holds that there
  exists a path $\aconf'' \rightarrow^\ast \aconf'''$ with
  $\state{\aconf'''}(q_f) > 0$.

  By definition of $K$ as $K=\max \{ \state{\eta_i}(q) \mid q\in Q,\ 1\leq
  i\leq m\}$ and configurations $\eta_i$ as minimal elements for the
  upward-closure $\Pre^\ast(\ConfsF) =
  \upcl\set{\eta_1,\ldots,\eta_m}$, we have that $K \leq L$ in any
  case. Hence, for our algorithm to be correct, it suffices to consider the
  symbolic graph of index $L\cdot \Msize Q$, as presented in
  Lemma~\ref{lem:algo}, and to solve a reachability problem over this graph.
  Let us study the size of this graph. Its state space is $V =
  \mathbb{N}_{L\cdot \Msize Q}^{Q} \times 2^Q\times D$.
The multisets of~$\mathbb{N}_{L\cdot \Msize Q}^{Q}$ are essentially 
mappings $Q \to \left[0; L\cdot \Msize Q \right]$.
Hence, we have that:
\(
  \Msize V \leq (L\cdot \Msize Q + 1)^{\Msize Q} \cdot 2^{\Msize Q} \cdot
  \Msize D 
\),
which is doubly-exponential in both the state space of the protocol and the
size of the data alphabet (because $L$~is). Since reachability over directed
graphs lies in $\NLOGSPACE$~\cite{Sip97} with regard to the size of the graph,
we~obtain $\NEXPSPACE$-membership with regard to the size of the protocol.
Finally, by Savitch's theorem~\cite{Sip97}, we know that
$\NEXPSPACE=\EXPSPACE$, which concludes our proof.
\end{proof}

\subsection{\PSPACE-hardness for deciding cut-offs}
\label{sec-hardness}

Our proof is based on the encoding of a linear-bounded Turing machine~\cite{Sip97}:
we~build a register protocol for which there is a negative cut-off if, and
only~if, the machine reaches its final state~$\qhalt$ with the tape head
reading the last cell of the tape.

\begin{theorem}
  \label{thm:negative-cutoff}
Deciding the existence of a negative cut-off is \PSPACE-hard.
\end{theorem}

Write~$n$ for the size of the tape of the Turing machine. We~assume
(without loss of generality) that the machine is deterministic, and that it
accepts only if it ends in its halting state~$q_{\halt}$ while reading the
last cell of the tape. 
Our~reduction works as follows (see~Fig.~\ref{fig-lbTM}): some processes of our network
will first be assigned   an index~$i$ in~$[1;n]$ indicating the cell of the tape they~shall encode during
the simulation. The~other processes are stuck in the initial location, and
will play no role. The~state~$q$ and position~$j$ of the head of the Turing
machine are stored in the register. During the simulation phase, when a
process is scheduled to play, it~checks in the register whether the tape head
is on the cell it encodes, and in that case it performs the transition of the
Turing machine. If the tape head is not on the cell it~encodes, the process
moves to the target location (which we consider as the target for the
almost-sure 
reachability problem). Finally, upon seeing~$(q_{\halt},n)$ in the register,
all~processes move to a $(n+1)$-filter protocol~$\calF_{n+1}$ (similar to that
of Fig.~\ref{fig-filtern}) whose last location~$s_{n+1}$ is the aforementioned
target location.

\usetikzlibrary{snakes}
\begin{figure}[h]
\centering
\begin{tikzpicture}[xscale=1.1]
\begin{scope}
\draw (0,0) node[rond,vert,inner sep=0pt] (qinit) {} node {\pinit};
\draw[latex'-] (qinit.135) -- +(135:3mm);
\draw (-2,0) node[rond,rouge,inner sep=0pt] (qinit') {} node {$\pinit'$};
\draw (-2,-1.3)  node[rond,bleu,inner sep=0pt] (c1) {} node {$1,c_1$};
\draw (-1,-1.3)  node[rond,bleu,inner sep=0pt] (c2) {} node {$2,c_2$};
\draw (0,-1.3)  node[rond,bleu,inner sep=0pt] (c3) {} node {$3,c_3$};
\draw (1,-1.3) node {...};
\draw (2,-1.3)  node[rond,bleu,inner sep=0pt] (cn) {} node {$n,c_n$};
\draw (2,0) node[rond,gris,inner sep=0pt] (qsink) {} node {$\psink$};
\begin{scope}[-latex']
\everymath{\scriptstyle}
\draw (qinit) -- (qinit') node[midway,above] {$R(\#)$};
\draw (qinit) -- (qsink) node[midway,above] {$R(D\setminus\{\#\})$};
\draw (qinit) -- (c2) node[left,midway] {$R(\#)$};
\draw (qinit) -- (c3) node[right,midway] {$R(\#)$};
\draw (qinit) -- (cn) node[right,midway] {$R(\#)$};
\draw (qinit') -- (c1) node[left,midway] {$W(q_0,1)$};
\draw (qsink) edge[out=-25,in=25,looseness=5] (qsink);
\end{scope}
\end{scope}
\begin{scope}[yshift=-2.5cm,xshift=0cm,yscale=1.2]
\draw (-2,0)  node[rond,bleu,inner sep=1pt] (c'1) {} node {$1,\sigma$};
\draw (-2,-1)  node[rounded corners=3.5mm,bleu,inner sep=1pt,minimum size=7mm]
  (c'2) {$1,\sigma,q\phantom'$}; 
\draw (-2,-2)  node[rond,bleu,inner sep=1pt] (c'3) {} node {$1,\sigma'$};
\draw (2,0)  node[rond,bleu,inner sep=1pt] (c''1) {} node {$n,\sigma'$};
\draw (2,-1)  node[rounded corners=3.5mm,bleu,inner sep=1pt,minimum size=7mm]
  (c''2) {$n,\sigma',q''$}; 
\draw (2,-2)  node[rounded corners=3.5mm,bleu,inner sep=1pt,minimum size=7mm] (c''3) {$n,\sigma''$};
\begin{scope}[-latex']
\everymath{\scriptstyle}
\draw[-latex'] (c'1) -- (c'2) node[midway,right] {$R(q,1)$};
\draw[-latex'] (c'2) -- (c'3) node[midway,right] {$W(q',2)$};
\draw[-latex'] (c''1) -- (c''2) node[midway,left] {$R(q'',n)$};
\draw[-latex'] (c''2) -- (c''3) node[midway,left] {$W(q,n-1)$};
\end{scope}
\draw[dashed] (c'1.-135) -- +(-135:3mm);
\draw[dashed] (c'1.-112.5) -- +(-112.5:3mm);
\draw[dashed] (c'1.-45) -- +(-45:3mm);
\draw[dashed] (c''1.-135) -- +(-135:3mm);
\draw[dashed] (c''1.-67.5) -- +(-67.5:3mm);
\draw[dashed] (c''1.-45) -- +(-45:3mm);
\draw[dashed] (c'3.-135) -- +(-135:3mm);
\draw[dashed] (c'3.-112.5) -- +(-112.5:3mm);
\draw[dashed] (c'3.-67.5) -- +(-67.5:3mm);
\draw[dashed] (c'3.-45) -- +(-45:3mm);
\draw[dashed] (c'3.-90) -- +(-90:3mm);
\draw[dashed] (c''3.-135) -- +(-135:3mm);
\draw[dashed] (c''3.-67.5) -- +(-67.5:3mm);
\draw[dashed] (c''3.-112.5) -- +(-112.5:3mm);
\draw[dashed] (c''3.-45) -- +(-45:3mm);
\draw[dashed] (c''3.-90) -- +(-90:3mm);
\foreach \i in {1,2,3,n}
  {\draw[dashed] (c\i.-135) -- +(-135:3mm);
   \draw[dashed] (c\i.-112.5) -- +(-112.5:3mm);
   \draw[dashed] (c\i.-67.5) -- +(-67.5:3mm);
   \draw[dashed] (c\i.-45) -- +(-45:3mm);
   \draw[dashed] (c\i.-90) -- +(-90:3mm);}

\end{scope}
\begin{scope}[yshift=-6.4cm,xshift=-0cm]
\draw (-3.5,0) node[rond,jaune,inner sep=0pt] (s0) {} node {$s_0$};
\draw (-2,0) node[rond,jaune,inner sep=0pt] (s1) {} node {$s_1$};
\draw (-1,0) node[rond,jaune,inner sep=0pt] (s2) {} node {$s_2$};
\draw (2,0) node[rond,jaune,inner sep=0pt,double] (sn1) {} node {$s_{n}$};
\draw (3.5,0) node[rond,jaune,inner sep=0pt,double] (sn) {} node {$s_{n+1}$};
\begin{scope}[-latex']
\everymath{\scriptstyle}
\draw (s0) edge[out=-160,in=-200,looseness=4] node[midway,left] {$W(f_0)$} (s0);
\draw (s0) edge[bend right] node[midway,below] {$R(f_0)$} (s1);
\draw (s1) edge[bend right] node[midway,below] {$W(f_1)$} (s0);
\draw (s1) edge[bend right] node[midway,below] {$R(f_1)$} (s2);
\draw (s2) edge[out=150,in=40] node[pos=.1,above] {$W(f_2)$} (s0);
\draw[-,dashed] (s2) edge[bend right] node[midway,below] {$R(f_2)$} +(.8,-.2);
\draw[-,dashed] ($(sn1)+(-.8,-.2)$) edge[bend right,-latex']
node[midway,below] {$R(f_{n-1})$} (sn1); 
\draw[-] (sn1) edge[bend right,-latex'] node[midway,below] {$R(f_n)$} (sn);
\draw (sn) edge[out=-20,in=20,looseness=4] (sn);
\end{scope}
\end{scope}
\begin{scope}
\everymath{\scriptstyle}
\draw[-latex',rounded corners=1mm] (c'1) -| (s0)
  node[pos=.9,left] {$\genfrac{}{}{0pt}{1}{R(\qhalt,n)}{R(f_i), i\in[0;n]}$};
\draw[-latex',rounded corners=1mm] (c'3) -| (s0);
\draw[-latex',rounded corners=1mm] (c''1) -| (sn) node[pos=.95,right] 
  {$\genfrac{}{}{0pt}{1}{R(\cdot,j), j\not=n}{R(\#)}$};
\draw[-latex',rounded corners=1mm] (c''3) -| (sn);
\draw[dashed,-latex'] (c'1) -- +(.7,0) node[right,font=\small] {(to~$s_{n+1}$)};
\draw[dashed,-latex'] (c'3) -- +(.7,0) node[right,font=\small] {(to~$s_{n+1}$)};
\draw[dashed,-latex'] (c''1) -- +(-.7,0) node[left,font=\small] {(to~$s_{0}$)};
\draw[dashed,-latex'] (c''3) -- +(-.7,0) node[left,font=\small] {(to~$s_{0}$)};
\end{scope}
\begin{scope}[xshift=-7mm]
\draw[snake=brace] (5,0.4) -- (5,-.8) node[midway,right=3mm,font=\small] {initialization phase};
\path[use as bounding box] (5,0);
\draw[snake=brace] (5,-1.3) -- (5,-5.2) node[midway,right=3mm,text width=4cm,font=\small]
     {simulation phase\par
       (for transitions \par $(q,\sigma) \to (q',\sigma',+1)$ and \par
      $(q'',\sigma') \to (q,\sigma'',-1)$)};
\draw[snake=brace] (5,-6) -- (5,-6.8) node[midway,right=3mm,text width=4.2cm,font=\small]
     {counting phase};
   \end{scope}
 \end{tikzpicture}
\caption{Distributed protocol $\calP_\calM$ encoding the linear-bounded Turing
  machine~$\calM$.} 
\label{fig-lbTM}
\end{figure}

If the Turing machine halts, then the corresponding run can be mimicked with
exactly one process per cell, thus giving rise to a finite run of the
distributed system where $n$ processes end up in the $(n+1)$-filter (and the
other processes are stuck in the initial location); from there $s_{n+1}$ cannot
be reached.
If the Turing machine does not halt, then assume that there is an infinite run
of the distributed system never reaching the target location. This run cannot
get stuck in the simulation phase forever, because it would end~up in a
strongly connected component from which the target location is reachable. Thus
this run eventually reaches the $(n+1)$-filter, which requires that at
least~$n+1$ processes participate in the simulation (because with $n$
processes it would simulate the exact run of the machine, and would not
reach~$q_{\halt}$, while with fewer processes the tape head could not go over
cells that are not handled by a process). Thus at~least $n+1$ processes would end up in
the $(n+1)$-filter, and with probability~$1$ the target location should be
reached. 

\smallskip
We~now formalize this construction, by describing the states and
transitions of the protocol within these three phases.  We fix a
linear-bounded Turing machine $\mathcal{M} =
(Q,q_0,\qhalt,\Sigma,\delta)$, where $Q$ is the set of states, $q_0,
\qhalt \in Q$ are the initial and halting states, $\Sigma$ is the
alphabet, and $\delta \subseteq Q \times \Sigma \times Q \times \Sigma
\times \{-1,+1\}$ is the set of transitions.
We~define the data alphabet $D=\{\#\}\uplus Q\times\Sigma \uplus \{f_i\mid
  0\leq i\leq n\}$, and the set of locations $P=\{\pinit,\pinit',\psink\}
  \uplus \bigl([1;n]\times \Sigma\times (Q\cup\{\epsilon\})\bigr) \uplus \{s_i\mid 0\leq i\leq
  n+1\}$. The~set of locations corresponds to three phases (see~Fig.~\ref{fig-lbTM}):
\begin{itemize}
\item The initialization phase contains $\pinit$, $\pinit'$ and $\psink$.
  From~the initial state~$\pinit$, upon reading~$\#$ (the~initial content of
  the register), the protocol has
  transitions to each state~$(i,\sigma_i)$ for all~$2\leq i\leq n$, where $\sigma_i$
  is the $i$-th letter of the initial content of the tape. If~reading anything
  different from~$\#$, the protocol moves to the sink state~$\psink$. Finally,
  there are transitions 
  $(\pinit,r(\#),\pinit')$ and $(\pinit',w(q_0,1),(1,c_1))$, where $q_0$ is the
  initial state of the Turing machine: this~pair of transitions is used to
  initialize the computation, by setting the content of the first cell and
  modifying the register, so that the initialization phase is over
  (there are no transitions writing~$\#$ in the register).
\item The second phase, called \emph{simulation phase}, uses register
  alphabet $Q \times [1,n]$, in~order to encode the state and position of the
  head of the Turing machine. The~state space for the simulation phase is
  $[1;n]\times \Sigma\times (Q\cup\{\epsilon\})$:
  state~$(i,\sigma,\epsilon)$ (written $(i,\sigma)$ in the sequel) encodes the
  fact that the content of the 
  $i$-th cell is~$\sigma$; the~states of the form $(i,\sigma,q)$ are
  intermediary states used during the simulation of one transition:
  when in state~$(i,\sigma)$ and reading~$(q,i)$ in the
  register, the~protocol moves to~$(i,\sigma,q)$, from which it~moves
  to~$(i,\sigma')$ and writes~$(q',j)$ in the register,
  provided that the machine has a transition
  $(q,\sigma)\to(q',\sigma',j-i)$.  If~the active process does not
  encode the position that the tape head is reading (i.e., the process
  is in state~$(i,\sigma)$ and reads~$(q,j)$ with~$j\not=i$)
  then it moves to the final state~$s_{n+1}$ of the third phase.
\item The role of the \emph{counting phase} is to count the number of processes
  participating in the simulation. When seeing the halting state in the
  register, each protocol moves to a module whose role is to check whether at
  least $n+1$ protocols are still ``running''. This uses data $\{f_i \mid
  0\leq i \leq n\}$ and states $\{s_i \mid i\in[0,n+1]\}$, with
  transitions from any state of the simulation phase to $s_0$ if the register
  contains $(\qhalt,n)$ or any of $\{f_i \mid 0\leq i\leq n\}$.
\end{itemize}

We~now prove that our construction is correct:
\begin{lemma}
  The register protocol $\calP_\calM$, with initial register content~$\#$ and
  target location~$s_{n+1}$, has a negative cut-off if, and only if, the Turing
  machine~$\mathcal{M}$ reaches~$\qhalt$ in the last cell of the tape. 
\end{lemma}

\begin{proof}
  First assume that there is a negative cut-off: there~exists $N_0$
  such that for any $N\geq N_0$, starting from the initial
  configuration~$\tuple{\pinit^N,\#}$ of the system~$\Sys_{\Prot_\calM}^{N}$ made of $N$~copies
  of~$\Prot_\calM$, the~probability that at least one process
  reaches~$s_{n+1}$ is strictly less than~$1$. 
  Since $\Sys_{\Prot_\calM}^{N}$ is a finite Markov
  chain, this implies that there is a cone of executions never
  visiting~$s_{n+1}$, i.e., a~finite execution~$\rho$ whose
  continuations never visit~$s_{n+1}$. 
  Since the register initially
  contains~$\#$, this finite execution (or~a~finite continuation of~it) must contain at least
  one configuration where some process has entered the simulation part.

  Now, in the simulation phase, we notice that, right after taking a transition
  $((i,\sigma,q),\penalty1000\relax w(q',{i\pm 1}),\penalty1000\relax (i,\sigma'))$, the~transition
  $((i,\sigma'),r(\cdot,j),s_{n+1})$ is always enabled. It~follows that
  at the end of the finite run~$\rho$, no~simulation transition should be
  enabled; hence all processes that had entered the simulation part
  should have left~it. Hence some process must have visited~$s_0$
  along~$\rho$ (because we~assume that $\rho$ does not involve~$s_{n+1}$).
  Moreover, by~Lemma~\ref{lem-filter}, 
  for~$s_{n+1}$ not to be reachable along 
  any continuation of~$\rho$, no~more than $n$ processes must be able to
  reach~$s_0$ along any continuation of~$\rho$, hence at most~$n$ processes
  may have entered the simulation phase.
  On~the other hand, for~$s_0$ to be visited, some process has to first write
  $(\qhalt,n)$ in the register; since the register initially
  contains~$(q_0,1)$, and no process can write $(\cdot{},i+1)$ without
  first reading~$(\cdot{},i)$, then for each $i\in[1,n]$ there must be at
  least one process visiting some state~$(i,\sigma_i)$, for some~$\sigma_i$;
  It~follows that at least $n$~processes must have entered the simulation
  phase. 

  In~the end, along~$\rho$, exactly one process visits $(i,c_i)$, for each
  $i\in [1,n]$, and encode the content of the $i$-th cell. As~a consequence,
  along~$\rho$, each cell of the tape of the Turing machine is encoded by
  exactly one process, and the execution mimics the exact computation of the
  Turing machine. Since the configuration~$(\qhalt,n)$ is eventually reached,
  the Turing machine halts with the tape head on the last cell of the
  tape.

  Conversely, assume the Turing machine halts, and consider the
  execution of $N\geq n$ processes where exactly one process goes in
  each of the $(i,c_i)$ and mimics the run of the Turing
  machine (the~other processes going to~$\psink$). We~get a finite execution
  ending up in a configuration 
  where all processes are either in~$\pinit$ or in~$\psink$, except
  for $n$ processes that are in the counting phase.
  No~continuation of this prefix ever reaches~$s_{n+1}$, so that the
  probability that some process reaches~$s_{n+1}$ is strictly less
  than~$1$.
\end{proof}

\section{Bounds on cut-offs}
\label{sec-bounds}

\subsection{Existence of exponential tight negative cut-offs}
We exhibit a family of register protocols that admits  negative
cut-off exponential in the size of the protocol. The~construction reuses ideas from the \PSPACE-hardness proof. Our
register protocol has two parts: one part simulates a counter over $n$ bits,
and requires a \emph{token} (a~special value in the register) to perform each
step of the simulation. The~second part is used to generate the tokens
(i.e.,~writing~$1$ in the register). Figure~\ref{fig-expo} depicts our
construction.
We~claim that this protocol, with~$\#$
as initial register value and~$q_f$ as~target location, admits a
negative tight cut-off larger than~$2^n$: in~other terms, there
exists~$N>2^n$ such that the final state will be reached with probability
strictly less than~$1$ in the
distributed system made of at least~$N$ processes (starting with~$\#$ in the
register), while the distributed system with $2^n$ processes will reach the
final state almost-surely.
In~order to justify this claim, we explain now the intuition behind this protocol.

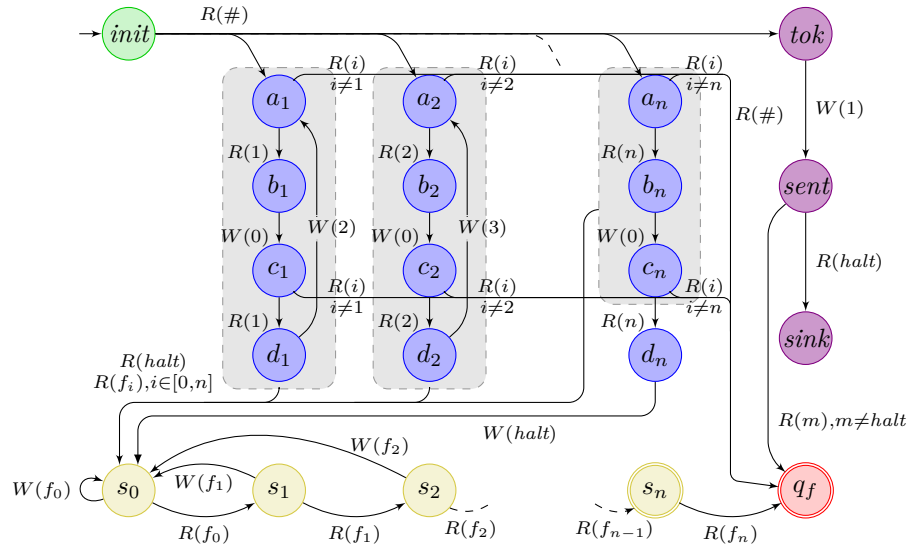
\begin{figure}[htb]
\centering
\begin{tikzpicture}[yscale=.9]
\draw (0,0.25)  node[rond,vert,inner sep=0pt] (a) {$\init$};
\draw[latex'-] (a.180) -- +(180:3mm);
\draw (9,0.25) node[rond,violet,inner sep=0pt] (b) {$\token$};
\draw (9,-2) node[rond,violet,inner sep=0pt] (c) {$\sent$};
\draw (9,-4.25) node[rond,violet,inner sep=0pt] (d) {$\sink$};
\draw[-latex'] (a) -- (b);
\draw[-latex'] (b) -- (c) node[midway,right] {$\scriptstyle W(1)$};
\draw[-latex'] (c) -- (d) node[midway,right] {$\scriptstyle R(\halt)$};
\begin{scope}[xshift=2cm,yshift=-7.5mm]
\draw[gris,dashed,rounded corners=2mm] (5,.5) -| +(.6,-3.5) --
  +(-.75,-3.5) |- +(0,0) node[pos=.2,coordinate] (yn) {};
\foreach \i/\j/\n in {0/2/1,2/3/2,5/n}
  {%
    \if n\n\else
      \draw[gris,dashed,rounded corners=2mm] (\i,.5) -| +(.75,-4.75) --
        +(-.75,-4.75) node[midway,coordinate] (z\n) {} |- +(0,0);\fi
    \draw (\i,0) node[rond,bleu,inner sep=0pt] (a\n) {$a_{\n}$};
    \draw (\i,-1.25) node[rond,bleu,inner sep=0pt] (b\n) {$b_{\n}$};
    \draw (\i,-2.5) node[rond,bleu,inner sep=0pt] (c\n) {$c_{\n}$};
    \draw (\i,-3.75) node[rond,bleu,inner sep=0pt] (d\n) {$d_{\n}$};
    \everymath{\scriptstyle}
    \draw[-latex'] (a\n) -- (b\n) node[pos=.8,left] {$R(\n)$};
    \draw[-latex'] (b\n) -- (c\n) node[pos=.8,left] {$W(0)$};
    \draw[-latex'] (c\n) -- (d\n) node[pos=.8,left] {$R(\n)$};
    \if n\n\else
    \draw[-latex',rounded corners=2mm] (d\n) -- ++(.5,.5) --
      ($(a\n)+(.5,-.5)$) node[midway,fill=fgris,minimum height=3.5mm] {} 
      node[midway,right=-2.5mm] {$W(\j)$} -- (a\n);\fi
  }
\draw[-latex',rounded corners=2mm] (a) -- ($(an)+(-.5,1)$) -- (an);
\draw[-latex',rounded corners=2mm] (a) -- ($(a2)+(-.5,1)$) -- (a2);
\draw[dashed,rounded corners=2mm] (a) -- ($(4,0)+(-.5,1)$) -- (3.75,.5);
\draw[-latex',rounded corners=2mm] (a) -- ($(a1)+(-.5,1)$) 
  node[pos=.8,above] {$\scriptstyle R(\#)$} -- (a1);
\end{scope}
\begin{scope}[yshift=-6.5cm,xshift=0cm]
\draw (0,0) node[rond,jaune,inner sep=0pt] (s0) {} node {$s_0$};
\draw (2,0) node[rond,jaune,inner sep=0pt] (s1) {} node {$s_1$};
\draw (4,0) node[rond,jaune,inner sep=0pt] (s2) {} node {$s_2$};
\draw (7,0) node[rond,jaune,inner sep=0pt,double] (sn1) {} node {$s_{n}$};
\draw (9,0) node[rond,rouge,inner sep=0pt,double] (sn) {} node {$q_f$};
\begin{scope}[-latex']
\everymath{\scriptstyle}
\draw (s0) edge[out=-160,in=-200,looseness=4] node[midway,left] {$W(f_0)$} (s0);
\draw (s0) edge[bend right] node[midway,below] {$R(f_0)$} (s1);
\draw (s1) edge[bend right] node[midway,below] {$W(f_1)$} (s0);
\draw (s1) edge[bend right] node[midway,below] {$R(f_1)$} (s2);
\draw (s2) edge[out=150,in=40] node[pos=.1,above] {$W(f_2)$} (s0);
\draw[-,dashed] (s2) edge[bend right] node[midway,below] {$R(f_2)$} +(.8,-.2);
\draw[-,dashed] ($(sn1)+(-.8,-.2)$) edge[bend right,-latex']
node[midway,below] {$R(f_{n-1})$} (sn1); 
\draw[-] (sn1) edge[bend right,-latex'] node[midway,below] {$R(f_n)$} (sn);
\end{scope}
\end{scope}
\draw[-latex',rounded corners=2mm] (c) -- +(-.5,-.5) --
  ($(sn)+(-.5,.5)$)  node[pos=.85,right] {$\scriptstyle R(m), m\not=\halt$} -- (sn);

\draw[-latex',rounded corners=1mm] (a1) -- +(.25,.4) -- ($(a1)+(.9,.4)$)
  node {$\genfrac{}{}{0pt}{1}{R(i)}{i\not=1}$} -- ($(an)+(1,.4)$)
   -- +(0,-6) node[pos=.1,right=-2pt] {$\scriptstyle R(\#)$} -- (sn);
\draw[rounded corners=1mm] (a2) -- +(.25,.4) -- ($(a2)+(.9,.4)$)
  node {$\genfrac{}{}{0pt}{1}{R(i)}{i\not=2}$} -- ($(an)+(.5,.4)$);
\draw[rounded corners=1mm] (an) -- +(.25,.4) -- ($(an)+(.65,.4)$)
  node {$\genfrac{}{}{0pt}{1}{R(i)}{i\not=n}$} -- ($(an)+(.9,.4)$);
\draw[rounded corners=1mm] (c1) -- +(.25,-.4) -- ($(c1)+(.9,-.4)$)
  node {$\genfrac{}{}{0pt}{1}{R(i)}{i\not=1}$} -- ($(cn)+(1,-.4)$) -- +(0,-.5);
\draw[rounded corners=1mm] (c2) -- +(.25,-.4) -- ($(c2)+(.9,-.4)$)
  node {$\genfrac{}{}{0pt}{1}{R(i)}{i\not=2}$} -- ($(cn)+(.5,-.4)$);
\draw[rounded corners=1mm] (cn) -- +(.25,-.4) -- ($(cn)+(.65,-.4)$)
  node {$\genfrac{}{}{0pt}{1}{R(i)}{i\not=n}$} -- ($(cn)+(.5,-.4)$);
\draw[-latex',rounded corners=2mm] (z1) |- +(-1,-.2) -| (s0.110)
  node[pos=0.3,above] {$\genfrac{}{}{0pt}{1}{R(\halt)}{R(f_i), i\in[0,n]}$};
\draw[rounded corners=2mm] (z2) |- +(-3,-.2);
\draw[rounded corners=2mm] (yn) -| +(-.2,-1) |- ($(z2)+(-1,-.2)$);
\draw[-latex,rounded corners=2mm] (dn) |- ($(z2)+(0,-.4)$) node[below,pos=.8]
  {$\scriptstyle W(\halt)$} -| (s0.70);

\end{tikzpicture}
\caption{Simulating an exponential counter: 
  grey boxes contain the nodes used to encode the bits of the counter; 
  yellow nodes at the bottom correspond to the filter module from
  Fig.~\ref{fig-filtern}; 
  purple nodes \token, \sent and \sink correspond to the second part of the
  protocol, and are used to produce tokens.
  Missing \emph{read} edges are assumed to be self-loops.
}
\label{fig-ex-exp2}\label{fig-expo}
\end{figure}

We first focus on the first part of the protocol, containing nodes
named~$a_i$, $b_i$, $c_i$, $d_i$ and~$s_i$. This part can be divided into
three phases: the~initialization phase lasts as long as the register
contains~$\#$; the counting phase starts when the register 
contains~$\halt$ for the first time; the~simulation phase is the intermediate phase.

During the initialization phase, processes move to locations~$a_i$
and~$\token$, until some process in~\token writes~$1$ in the register
(or~until some process reaches~$q_f$, using a transition from~$a_i$ to~$q_f$
while reading~$\#$). 
Write~$\aconf_0$ for the configuration reached when entering the simulation
phase (i.e., when~$1$ is written in the register for the first time).
We~assume that $\state{\aconf_0}(a_i)>0$ for some~$i$, as otherwise all the
processes are in~$\token$, and they all will eventually reach~$q_f$. Now, we
notice that if $\state{\aconf_0}(a_i)=0$ for some~$i$, then location~$d_n$ cannot
be reached, so that no process can reach the counting phase. In~that case,
some process (and~actually all of~them) will eventually reach~$q_f$.
We~now~consider the case where $\state{\aconf_0}(a_i)\geq 1$ for all~$i$.
One~can
prove (inductively) that $d_i$~is reachable when
$\state{\aconf_0}(\token)\geq 2^i$. Hence $d_n$, and thus also~$s_0$, can be reached when
$\state{\aconf_0}(\token)\geq 2^n$.
Assuming $q_f$ is not reached, the counting phase must never contain more than
$n$ processes, hence we actually have that $\state{\aconf_0}(a_i)= 1$.
With this new condition, $s_0$ is reached if, and only~if,
$\state{\aconf_0}(\token)\geq 2^n$.
When the latter condition is not true,
$q_f$ will be reached almost-surely, which proves the second part of our
claim: the~final location is reached almost-surely in systems with strictly
less than~$n+2^n$ copies of the protocol. 

We now consider the case of systems with at least $n+2^n$ processes.
We~exhibit a finite execution of those systems from which no continuation can
reach~$q_f$, thus proving that $q_f$ is reached with probability strictly less
than~$1$ in those systems. The~execution is as follows: during initialization,
for each~$i$, one process enters~$a_i$; all other processes move to~$\token$,
and one of them write~$1$ in the register. The $n$ processes in the simulation
phase then simulate the consecutive incrementations of the counter, consuming
one token at each step, until reaching~$d_n$. At~that time, all the processes
in~\token move to~\sent, and the process in~$d_n$ writes~$\halt$ in the
register and enters~$s_0$. The processes in the simulation phase can then
enter~$s_0$, and those in~$\sent$ can move to~$\sink$. We~now have $n$
processes in~$s_0$, and the other ones in~\sink. According to
Lemma~\ref{lem-filter}, location~$q_f$ cannot be reached from this
configuration, which concludes our proof.

\begin{theorem}
There exists a family of register protocols which, equipped with an
initial register value and a target location, admit negative tight cut-offs
whose  size are exponential in the size of the protocol.
\end{theorem}

\begin{remark}
The question whether there exists protocols with exponential 
\emph{positive} cut-offs remains open. The family of \emph{filter} protocols
described at Section~\ref{subsec:example} is an example of protocols
with a 
linear positive cut-off.
\end{remark}

\subsection{Upper bounds on tight cut-offs}

The results (and proofs) of Section~\ref{sec-algo} 
can be used to derive upper bounds on tight cut-offs. 
We~make this explicit in the~following theorem.

\begin{theorem}
\label{thm:bounds}
  For a protocol $\Prot=\tuple{Q,D,q_0,T}$ equipped with an initial register
  value~$d_0\in D$ and a target location~$q_f\in Q$, the tight cut-off
  is at most doubly-exponential in $\Msize \Prot$.
\end{theorem}

\begin{proof}
First assume that the cut-off is negative. From~Lemma~\ref{lem:algo}, 
there is a
state~$\tuple{\mu,S,d}$ in the symbolic graph of index~$K\cdot |Q|$ that is
reachable from~$(q_0^{(K\cdot |Q|)},\{q_0\},d_0)$ and from which no
configutation containing $q_f$ is reachable. Applying Lemma~\ref{lemma-gsymb},
there exist multisets~$\delta_0=q_0^N$ and~$\delta$, with respective
supports $\{q_0\}$ and~$S$, such that $\tuple{\mu\Mcup\delta,d}$ is reachable
from~${\tuple{q_0^{K\cdot |Q|}\Mcup\delta_0,d_0}}$. 
Hence $N+K\cdot |Q|$ is a negative cut-off.

Let us evaluate the size of~$N$: this number is extracted from
the symbolic path $\tuple{q_0^{(K\cdot |Q|)},\{q_0\},d_0}\rightarrow^*
\tuple{\mu,S,d}$, which has length at most $|V|-1$ (where $V$ is the set of
states of the symbolic graph of index~$K\cdot |Q|$). 
Applying Lemma~\ref{lem:algo} $|V|-1$ times, increasing the size of the
concrete representation of $S$ by one each time, we~get $N\leq |V|$. Thus, both
$K\cdot |Q|$ and $N$ are doubly-exponential in~$\Msize\Prot$, 
thanks to the proof of Theorem~\ref{thm:expspace-membership}.

The proof of Lemma~\ref{lem:algo} also entails that if the distributed system
with some $N>K\cdot |Q|$ processes does not 
almost-surely reach the target state, then there is a negative cut-off. Hence
for there to be a positive cut-off, the target has to be almost-surely
reachable for all~$N>K\cdot |Q|$, which makes $K\cdot |Q|$ a
  (doubly-exponential) positive cut-off.
\end{proof}

\section{Conclusions and future works}

We have shown that in networks of identical finite-state automata
communicating (non-atomically) through a single register and equipped
with a fair stochastic scheduler, there always exists a cut-off on the
number of processes which either witnesses almost-sure reachability of
a specific control-state (positive cut-off) or its negation (negative
cut-off). This cut-off determinacy essentially relies on the
monotonicity induced by our model, which allows to use well-quasi
order techniques. By analyzing a well-chosen symbolic graph, one can
decide in \EXPSPACE whether that cut-off is positive, or negative, and
we proved this decision problem to be \PSPACE-hard. This approach allows us to deduce some
doubly-exponential bounds on the value of the cut-offs. Finally, we
gave an example of a network in which there is a negative cut-off,
which is exponential in the size of the underlying protocol. Note
however that no such lower-bound is known yet for positive cut-offs.

We have several further directions of research. First, it would be
nice to fill the gap between the \PSPACE lower bound and the \EXPSPACE
upper bound for deciding the nature of the cut-off.  
We would like also to investigate further atomic read{\slash}write
operations, which generate non-monotonic transition systems, but for
which we would like to decide whether there is a cut-off or not.
Finally, we believe that our techniques could be extended to more
general classes of properties, for instance, universal reachability
(all processes should enter a distinguished state), or liveness
properties.


\end{document}